\documentclass[11pt]{article}

\usepackage{amsthm}
\usepackage{amsmath}
\usepackage{graphicx} 
\usepackage{array} 

\usepackage{amsmath, amssymb, amsfonts, verbatim}
\usepackage{hyphenat,epsfig,subfigure,multirow}

\usepackage[usenames,dvipsnames]{xcolor}

\usepackage{color}

\usepackage{tcolorbox}
\tcbuselibrary{skins,breakable}
\tcbset{enhanced jigsaw}

\usepackage[compact]{titlesec}

\definecolor{DarkRed}{rgb}{0.5,0.1,0.1}
\definecolor{DarkBlue}{rgb}{0.1,0.1,0.5}

\usepackage[linktocpage=true,
	pagebackref=true,colorlinks,
	linkcolor=DarkRed,citecolor=ForestGreen,
	bookmarks,bookmarksopen,bookmarksnumbered]
	{hyperref}

\usepackage[lined,boxed,ruled,norelsize,algo2e,linesnumbered]{algorithm2e}

\usepackage{bm}
\usepackage{url}
\usepackage{xspace}
\usepackage[mathscr]{euscript}

\usepackage{mdframed}

\usepackage[noend]{algpseudocode}
\makeatletter
\def\BState{\State\hskip-\ALG@thistlm}
\makeatother

\usepackage{cite}
\usepackage{enumitem}

\usepackage{stmaryrd}

\usepackage[margin=1in]{geometry}

\newtheorem{theorem}{Theorem}
\newtheorem{lemma}{Lemma}[section]
\newtheorem{proposition}[lemma]{Proposition}
\newtheorem{corollary}[lemma]{Corollary}
\newtheorem{claim}[lemma]{Claim}
\newtheorem{fact}[lemma]{Fact}

\newtheorem{definition}{Definition}

\newtheorem{problem}{Problem}
\newtheorem{remark}[lemma]{Remark}
\newtheorem{observation}[lemma]{Observation}

\newtheorem*{claim*}{Claim}
\newtheorem*{proposition*}{Proposition}
\newtheorem*{lemma*}{Lemma}
\newtheorem*{problem*}{Problem}
\newtheorem*{theorem*}{Theorem}

\newtheorem{mdresult}[theorem]{Theorem}
\newenvironment{result}{\begin{mdframed}[backgroundcolor=lightgray!40,topline=false,rightline=false,leftline=false,bottomline=false,innertopmargin=2pt]\begin{mdresult}}{\end{mdresult}\end{mdframed}}

\newtheorem{mdinvariant}{Invariant}

\allowdisplaybreaks

\newcommand{\R}{{\mathbb{R}}}

\newcommand{\cB}{\mathcal{B}}

\usepackage{tikz}
\newcommand*\circled[1]{\tikz[baseline=(char.base)]{
            \node[shape=circle,draw,inner sep=1pt] (char) {#1};}}
            
\DeclareMathOperator{\replacement}{\circled{{\rm r}}} 
\DeclareMathOperator{\zigzag}{\circled{{\rm z}}} 

\renewcommand{\qed}{\nobreak \ifvmode \relax \else
      \ifdim\lastskip<1.5em \hskip-\lastskip
      \hskip1.5em plus0em minus0.5em \fi \nobreak
      \vrule height0.75em width0.5em depth0.25em\fi}

\newcommand{\PG}{\ensuremath{\mathbb{G}}}

\newcommand{\Qed}[1]{\ensuremath{\qed_{\textnormal{~#1}}}}

\newcommand{\be}{\ensuremath{\bm{e}}}

\newcommand{\PP}{\ensuremath{\mathcal{P}}}

\newcommand{\FGS}{\ensuremath{\FG_{\textnormal{\textsf{S}}}}}

\newcommand{\SRW}{\ensuremath{\textnormal{\textsf{SimpleRandomWalk}}}\xspace}

\newcommand{\distrw}{\ensuremath{\dist_{\textnormal{\textsf{RW}}}}}

\newcommand{\distribution}[1]{\ensuremath{{\mathbb{D}}(#1)}}

\newcommand{\unif}{\ensuremath{\mathcal{U}}}

\newcommand{\bW}{\ensuremath{\overline{W}}}

\newcommand{\LL}{\ensuremath{\mathcal{L}}}

\newcommand{\GC}{\ensuremath{\textnormal{\textsf{GrowComponents}}}\xspace}

\newcommand{\CC}{\ensuremath{\mathcal{C}}}

\newcommand{\range}[1]{\ensuremath{\llbracket{#1}\rrbracket}}

\newcommand{\LE}{\ensuremath{\textnormal{\textsf{LeaderElection}}}\xspace}

\newcommand{\tG}{\ensuremath{\widetilde{G}}}

\newcommand{\beps}{\ensuremath{\overline{\eps}}}

\newcommand{\HH}{\ensuremath{\mathcal{H}}}
\newcommand{\ones}{\ensuremath{\bm{1}}}

\newcommand{\RGC}{\ensuremath{\textnormal{\textsf{RegularGraphConstruction}}}\xspace}
\newcommand{\PR}{\ensuremath{\textnormal{\textsf{ReplacementProduct}}}\xspace}

\newcommand{\bH}{\ensuremath{\overline{H}}}

\newcommand{\Mark}{\ensuremath{\textnormal{\textsf{Mark}}}}
\newcommand{\ParallelMark}{\ensuremath{\textnormal{\textsf{DetectIndependence}}}}

\newcommand{\gamstar}{\ensuremath{\gamma^{*}}}

\newcommand{\ExpConn}{\ensuremath{\textnormal{\textsf{ExpanderConn}}_n}}

\newcommand{\toShrink}{-.20cm}
\newcommand{\toShrinkEnu}{-.2cm}

\newcommand{\Leq}[1]{\ensuremath{\underset{\textnormal{#1}}\leq}}
\newcommand{\Geq}[1]{\ensuremath{\underset{\textnormal{#1}}\geq}}
\newcommand{\Eq}[1]{\ensuremath{\underset{\textnormal{#1}}=}}

\newcommand{\algline}{
  \rule{0.5\linewidth}{.1pt}\hspace{\fill}%
  \par\nointerlineskip \vspace{.1pt}
}

\newcommand{\jstar}{\ensuremath{{j^{\star}}}}

\newcommand{\tvd}[2]{\ensuremath{\card{#1 - #2}_{tvd}}}
\newcommand{\Ot}{\ensuremath{\widetilde{O}}}
\newcommand{\eps}{\ensuremath{\varepsilon}}
\newcommand{\Paren}[1]{\Big(#1\Big)}

\newcommand{\bracket}[1]{\left[#1\right]}
\newcommand{\paren}[1]{\ensuremath{\left(#1\right)}\xspace}
\newcommand{\card}[1]{\left\vert{#1}\right\vert}

\newcommand{\IR}{\ensuremath{\mathbb{R}}}

\newcommand{\norm}[1]{\ensuremath{\|#1\|}}

\newcommand{\floor}[1]{{\left\lfloor{#1}\right\rfloor}}

\newcommand{\set}[1]{\ensuremath{\left\{ #1 \right\}}}

\newcommand{\polylog}{\mbox{\rm  polylog}}

\DeclareMathOperator*{\Exp}{\ensuremath{{\mathbb{E}}}}
\DeclareMathOperator*{\Prob}{\ensuremath{\textnormal{Pr}}}
\renewcommand{\Pr}{\Prob}

\newcommand{\Ex}{\Exp}

\newcommand{\etal}{{\it et al.\,}}

\newcommand{\FG}{\ensuremath{\mathcal{G}}}
\newcommand{\FV}{\ensuremath{\mathcal{V}}}

\newcommand{\event}[1]{\ensuremath{{\sf E}_{#1}}}

\newenvironment{tbox}{\begin{tcolorbox}[
		enlarge top by=5pt,
		enlarge bottom by=5pt,
		 breakable,
		 boxsep=0pt,
                  left=4pt,
                  right=4pt,
                  top=10pt,
                  arc=0pt,
                  boxrule=1pt,toprule=1pt,
                  colback=white
                  ]
	}
{\end{tcolorbox}}

\newcommand{\dist}{\ensuremath{\mathcal{D}}}

\renewcommand{\event}{\mathcal{E}}

\title{Massively Parallel Algorithms for Finding Well-Connected Components in Sparse Graphs}
\author{Sepehr Assadi\footnote{Supported in part by NSF grant CCF-1617851.} \\ University of Pennsylvania \\ { \texttt{sassadi@cis.upenn.edu}} \and Xiaorui Sun \\ Microsoft Research\\ {\texttt{xiaorsun@microsoft.com}} \and Omri Weinstein \\ Columbia University \\ {\texttt{omri@cs.columbia.edu}}} 

\date{}

\begin{document}
\maketitle

\thispagestyle{empty}
\begin{abstract}

	A fundamental question that shrouds the emergence of \emph{massively parallel computing (MPC)} platforms is 
	how can the additional power of the MPC paradigm (more local storage and computational power) be leveraged 
	to achieve faster algorithms compared to classical parallel models such as PRAM? 
	
	\smallskip
	
	Previous research has identified the \emph{sparse graph connectivity} problem as a major obstacle to such improvement: 
	While classical logarithmic-round PRAM algorithms for finding connected components in any $n$-vertex graph have been known for more than three decades, 
	no $o(\log{n})$-round MPC algorithms are known for this task with truly sublinear in $n$ memory per machine. 
	This problem arises when processing massive yet sparse graphs with $O(n)$ edges, for which the interesting setting of parameters is $n^{1-\Omega(1)}$ memory per machine.  
	It is conjectured that achieving an $o(\log{n})$-round algorithm for connectivity on general sparse graphs with $n^{1-\Omega(1)}$ per-machine memory may not be possible, and   
	this conjecture also forms the basis for multiple \emph{conditional} hardness results on the round complexity of other problems in the MPC model.
	
	\smallskip
	
	In this paper, we take an \emph{opportunistic} approach towards the sparse graph connectivity problem, by designing an algorithm with improved performance guarantees 
	in terms of the connectivity structure of the input graph.  Formally, we design an MPC algorithm that finds all connected components with \emph{spectral gap} at least $\lambda$ in a graph  
	in $O(\log\log{n} + \log{(1/\lambda)})$ MPC rounds and  $n^{\Omega(1)}$ memory
	per machine. 
	While this algorithm still requires $\Omega(\log{n})$ rounds in the worst-case
	when components are ``weakly'' connected (i.e., $\lambda \approx 1/n )$, it achieves an \emph{exponential} round reduction on 
	sparse ``well-connected'' components (i.e., $\lambda \geq 1/\polylog{(n)}$) using only $n^{\Omega(1)}$ memory per machine and $\Ot(n)$ total memory, and still operates in $o(\log n)$ rounds even when $\lambda = 1/n^{o(1)}$.  
	To best of our knowledge, this is the first non-trivial (and indeed exponential) improvement in the round complexity over PRAM algorithms, for a natural class of sparse connectivity instances.  
	
\end{abstract}

\clearpage

\thispagestyle{empty}
\setcounter{tocdepth}{3}
\tableofcontents
\clearpage
\setcounter{page}{1}

\section{Introduction}\label{sec:lb}
Recent years have witnessed a resurgence of interest in the theory of parallel computation, motivated by the 
successful deployment of parallel computing platforms such as MapReduce, Hadoop and Spark \cite{DG04,ZCF10,Whi09}. 
The \emph{massively parallel computation (MPC)} model \cite{KarloffSV10,GoodrichSZ11,AndoniNOY14,BeameKS13} 
is a theoretical abstraction 
which aims to capture the design principles and main distinguishing aspects of these platform over the 
classical PRAM model: more local computation power (in principle unbounded), and larger local memory per processor. 
Consequently, it is typically possible to simulate a PRAM algorithm in the MPC model with no asymptotic blowup in the number of rounds \cite{KarloffSV10,GoodrichSZ11}. 
However, these powerful features anticipate a potential for qualitatively smarter and dramatically faster parallel algorithms. A fundamental question
on this front is then: 

\begin{quote}
\emph{%
	How can the additional power of the MPC model (more local storage and computational power) be leveraged 
	to achieve faster algorithms compared to classical parallel models such as PRAM algorithms? 
}
\end{quote}

The answer to this question turns to be highly dependent on the type of problems at hand and the setting of parameters. 
For graph problems---the focus of this paper---
the first improvement over PRAM algorithms was already achieved by Karloff ~\etal~\cite{KarloffSV10} 
who developed algorithms for graph connectivity and MST in $O(1)$ MPC rounds
on machines with local memory $n^{1+\Omega(1)}$; here, $n$ is the number of vertices in the graph. This is in contrast to the $\Omega(\log{n})$ round needed in the standard PRAM model for
these problems (see, e.g.,~\cite{Shiloach82logn, Reif85optimal, Gazit86optimal,CookDR86,ParberryY91,KNP92,HalperinZ94}). Since then, numerous algorithms have been designed for 
various graph problems that achieve $O(1)$ round-complexity with local memory $n^{1+\Omega(1)}$ on each machine (see, e.g.,~\cite{LattanziMSV11,KumarMVV13,AhnG15,AssadiK17} and references therein).

The next set of improvements reduced the memory per machine to $O(n)$ (possibly at the cost of a slight increase in the number of rounds). 
For example, an $O(1)$ round algorithm for MST and connectivity using only $O(n)$ memory per machine has
been proposed in~\cite{Jurdzinski018} building on previous work in~\cite{GhaffariP16,HegemanPPSS15,LotkerPPP03} (see also~\cite{AhnGM12Linear,BehnezhadDH18,Lenzen13} for related results). 
A series of very recent papers~\cite{Assadi17,AssadiBBMS17,CzumajLMMOS17,GhaffariGMR18,Konrad18}, initiated by a breakthrough result of~\cite{CzumajLMMOS17}, have also achieved an $O(\log\log{n})$-round algorithms for different graph problems
such as matching, vertex cover, and MIS in the MPC model, when the memory per machine is $O(n)$ or even $O(n/\polylog{(n)})$. 

Alas, this progress has came to a halt at the \emph{truly sublinear} in $n$ regime, i.e., $n^{1-\Omega(1)}$ space per-machine. 
This setting of parameter is particularly relevant to \emph{sparse} graphs with $O(n)$ edges, as in this scenario, 
$\Omega(n)$ memory per-machine allows to fit the entire input on a single machine, thereby trivializing the problem.   
We remark that sparse graphs arise in many practical large-scale networks, such as social networks, which are believed to have only $O(n)$ edges. 

The aforementioned line of research has identified a captivating algorithmic challenge for breaking the linear-memory barrier in the MPC 
model: \emph{connectivity on sparse undirected graphs}. 
While classic $O(\log n)$-round PRAM algorithms for connectivity in undirected 
graphs have been known for more than three decades (see~\cite{Shiloach82logn, Reif85optimal, Gazit86optimal,KNP92,HalperinZ94} 
and references therein), no faster MPC algorithm 
with truly sublinear $n^{1-\Omega(1)}$-memory per machine is known to date (see, e.g.~\cite{KarloffSV10,KiverisLMRV14,RastogiMCS13}). 

There are several substantial reasons for 
the lack of progress on this fascinating problem. On one hand, 
$\Omega(\log n)$ rounds are known to be necessary for a \emph{restricted} class of (routing-style) 
MPC algorithms \cite{BeameKS13}, and in fact it has been conjectured that this logarithmic barrier may be unavoidable for 
any MPC algorithm~\cite{AndoniNOY14,BeameKS13,RastogiMCS13,RoughgardenVW16}. This belief  
led to a series of recent results that used sparse connectivity as a hardness assumption for proving
\emph{conditional} lower bounds in the MPC model for other problems (see~\cite{AndoniNOY14,Yaroslavtsev2017massively} and references therein). 
On the other hand, it was observed by~\cite{RoughgardenVW16} that proving \emph{any $\omega(1)$} lower bound on the round complexity of this problem would imply $\mathbf{NC}^1 \subsetneq \mathbf{P}$, a major breakthrough in complexity theory  
which seems beyond the reach of current techniques.

In this paper we take an \emph{opportunistic} 
approach to the sparse connectivity problem, which exploits the connectivity structure of the underlying graph. In particular, 
we use \emph{spectral gap} as a quantitative measure of ``connectedness'' of a graph and design an algorithm for connectivity 
with improved performance guarantee depending on the spectral gap of the connected components of the underlying graph.
For example, when connected components of the graph have large spectral gap, say $\Omega(1)$ or even $\Omega(1/\polylog{(n)})$ (as in expanders), 
our algorithm only requires $O(\log\log{n})$ MPC rounds
while using $n^{\Omega(1)}$ memory per machine and $\Ot(n)$ total memory. To our knowledge, this constitutes the {first} non-trivial improvement on the standard $O(\log{n})$ round algorithms for connectivity in the MPC model when the memory
per machine is $n^{\Omega(1)}$ for a general family of input graphs. We elaborate more on our results in Section~\ref{sec:results}.


\paragraph{Massively Parallel Computation Model.}
We adopt the most stringent model of modern parallel computation among~\cite{KarloffSV10,GoodrichSZ11,AndoniNOY14,BeameKS13}, the so-called \emph{Massively Parallel Computation (MPC)} model of~\cite{BeameKS13}. 
Let $N$ be the input size. It is assumed that the local memory on each machine and the number of machine should be at most $N^{1-\delta}$ for some constant $\delta >0$. 
Additionally, we are typically interested in algorithms with total memory $\Ot(N)$, i.e., proportional to the input size (within logarithmic factors)~\cite{BeameKS13,AndoniNOY14}, although
total memories as large as $O(N^{2-2\delta})$ have been also considered in more relaxed variants~\cite{KarloffSV10}. 
 The motivation behind these constraints is that the number of machines, and local memory of each machine should be much smaller than the input size to the problem since these frameworks are used to process large datasets. 

In this model, computation proceeds in rounds. During a round each machine runs a local algorithm on the data assigned to the machine. No communication between machines is allowed during a round. Between rounds, machines are allowed to communicate so long as each machine send or receive a communication no more than its memory. Any data output from a machine must be computed locally from the data residing on the machine and initially the input data is distributed across machines adversarially. 
The goal is to minimize the total number of rounds. 


\subsection{Our Results and Techniques}\label{sec:results}

The \emph{spectral gap} of a graph is defined to be the second eigenvalue of the normalized Laplacian associated with this graph (see Section~\ref{sec:spectral-gap} for more details). We use spectral gap as a measure of ``connectedness'' 
of a graph and design an \emph{opportunistic} algorithm\footnote{We borrow the term of ``opportunistic'' algorithm from Farach-Colton and Thorup~\cite{FT98} which defined it in the context of string matching.} for connectivity with improved performance guarantee depending on the spectral gap of the underlying graph.  
\begin{result}[\textbf{Main Result}]\label{res:main}
	There exists an MPC algorithm that with high probability identifies all connected components of any given sparse undirected $n$-vertex graph $G(V,E)$ with 
	$\Ot(n)$ edges and a lower bound of $\lambda \in (0,1)$ on the spectral gap of its connected components.
	
	\noindent
	For constant $\delta > 0$, the algorithm can be implemented with $O(\frac{1}{\lambda^2} \cdot n^{1-\delta} \cdot \polylog{(n)})$ machines each with 
	$O(n^{\delta} \cdot \polylog{(n)})$ memory, and in ${O}({\log\log{n} + \log{\paren{{1}/{\lambda}}}})$ MPC rounds. 
\end{result}

Theorem~\ref{res:main} can be extended to the case when the algorithm is \emph{oblivious} to the value of $\lambda$ and still manages to achieve an improved performance depending on this parameter (see Section~\ref{sec:everything}).  
Our result is most interesting in the case when spectral gap of (each connected component) of the graph is lower bounded by a constant or even $1/\polylog{(n)}$, i.e., 
for graphs with ``well-connected''  components. Examples of such graphs include random graphs\footnote{This in particular means that in a probabilistic sense, this setting of parameter applies to almost all graphs.} and
expanders (see also~\cite{MalliarosM11,GkantsidisMS03} for real-life examples in social networks). In this case, we obtain an algorithm with $\Ot(n)$ total memory and $n^{\Omega(1)}$ memory per machine which 
can identify all connected components in only $O(\log\log{n})$ rounds. To our knowledge, this constitutes the {first} non-trivial improvement on the standard $O(\log{n})$ 
round algorithms for connectivity in the MPC model when the memory
per machine is $n^{\Omega(1)}$ for a general family of input graphs. 

Nevertheless, the algorithm in Theorem~\ref{res:main} still manages to achieve a non-trivial improvement even when the spectral gap is as small as $1/n^{o(1)}$. Even in this case, 
the algorithm requires $o(\log{n})$ MPC rounds (and total memory which is larger than the input size by only an $n^{o(1)}$ factor). This means that the algorithm benefits from the extra power of the MPC model (much more local computation power) compared to
the classical parallel algorithms in the PRAM model which require $\Omega(\log{n})$ rounds to solve connectivity (even on sparse expanders; see below). 

We also prove an \emph{unconditional} $\Omega(\log_s n)$-round lower bound for the promise problem of connectivity on \emph{sparse expanders} on machines with memory $s$. This implies that the ``full power" of the MPC model is indeed required to achieve 
our speedup, as with $s=\polylog{(n)}$ memory,  $\Omega(\log_s n)=\tilde{\Omega}(\log{n})$ rounds are needed \emph{even on sparse expanders} (this result, as well as a lower bound for PRAM algorithms and further 
discussion are presented in Section \ref{sec:lb}). We remark that by a result of~\cite{RoughgardenVW16}, our lower bound is the best 
possible unconditional lower bound short of proving that $\mathbf{NC}^1 \subsetneq \mathbf{P}$ which would be a major breakthrough in complexity theory. 

Finally, we note that a simple application of the toolkit we develop in proving our main result in Theorem~\ref{res:main} also implies that one can solve the
connectivity problem in only $O(\log\log{n})$ MPC rounds on \emph{any} graph (with no assumption on spectral gap, etc.) when the memory per-machine is \emph{mildly sublinear} in $n$, i.e.,
is $O(n/\polylog{(n)})$. Formally, 

\begin{theorem}\label{thm:sublinear}
	There exists an MPC algorithm that given any arbitrary $n$-vertex graph $G(V,E)$ with high probability identifies all connected components of $G$ in $O(\log\log{n} + \log{\paren{\frac{n}{s}})}$ MPC rounds
	 on machines of memory $s = n^{\Omega(1)}$. 
\end{theorem}

Theorem~\ref{thm:sublinear} is reminiscent of the recent set of results in~\cite{Assadi17,AssadiBBMS17,CzumajLMMOS17,GhaffariGMR18,Konrad18} on achieving similar guarantees for other graph problems
such as maximum matching and minimum vertex cover in the mildly sublinear in $n$ per-machine memory regime. This result emphasizes the \emph{truly sublinear} in $n$ regime, i.e., $n^{1-\Omega(1)}$ per-machine memory, as the 
``real'' barrier to obtaining efficient algorithms for sparse connectivity with improved performance compared to PRAM algorithms. 

\paragraph{Techniques.}  
The first main technical ingredient of this work  is a distributed ``data structure" for performing and processing short \emph{independent random walks} 
(proportional to the \emph{mixing time} of each component) from  
\emph{all} vertices of the graph simultaneously, whose construction takes \emph{logarithmic} number of rounds in length of the walk. 
While implementing random walks in distributed and parallel settings is a well-studied problem 
(see, e.g.,~\cite{KNP92,HalperinZ94,SarmaGP11,SarmaNPT13} and references therein), the guarantee of our algorithm in achieving \emph{independent} random 
walks across all vertices in a small number of rounds and total memory, departs from previous work (independence is crucial in the context of our algorithm). 
Achieving this stronger guarantee requires different tools, in particular, a method for ``regularizing'' our graph using a parallel implementation of the  
\emph{replacement product} operation (see, e.g.~\cite{ReingoldVW00}) that we design in this paper.  

Our second main technical ingredient is a novel \emph{leader-election} based algorithm for finding a spanning tree of a 
\emph{random graph}. The key feature of 
this algorithm that distinguishes it from previous MPC algorithms for sparse connectivity (see, e.g.,~\cite{KiverisLMRV14,RastogiMCS13,KarloffSV10}) is that 
on random graphs, it provably requires only $O(\log\log{n})$ MPC rounds as opposed to $\Omega(\log{n})$ (we point out that~\cite{KiverisLMRV14} also analyzed their algorithms on random graphs (see Lemma~9), but even on random graphs their algorithm
requires $\Theta(\log^{2}{n})$ rounds). Our algorithm achieves this \emph{exponential} speedup by contracting \emph{quadratically} 
larger components to a single vertex in each step, while ``preserving the randomness'' in the resulting contracted graph
to allow for recursion. 
We elaborate on our techniques in the streamlined overview of our algorithm in Section~\ref{sec:technical}.

\begin{remark}\label{rem:loglogn}  
	
\emph{
	Our techniques and those in the very recent series of $O(\log\log{n})$-round MPC algorithms for various graphs problems~\cite{Assadi17,AssadiBBMS17,CzumajLMMOS17,GhaffariGMR18,Konrad18}
	are \textbf{\emph{entirely different}} and in some sense, ``dual'' to each other: 
	The latter results are all at their core based on ``sparsifying'' the input graph in successive rounds by decreasing its maximum degree by a quadratic factor. Our leader-election based algorithm on the other hand ``densifies'' the graph over rounds by increasing the degree of 
	all vertices by a quadratic factor (and decreasing the number of vertices in the process as well). We further emphasize that all these previous results are stuck in a crucial way at memory of $n^{1-o(1)}$ and with 
	$n^{1-\Omega(1)}$ memory---the main focus of our paper and the setting of interest for sparse connectivity---their performance degrade to $\Omega(\log{n})$ rounds. }
\end{remark}

\subsection{Further Related Work}\label{sec:related}

Finding connected components in undirected graphs has been studied extensively in the 
MPC model~\cite{KarloffSV10,LattanziFMS11,Afrati2011map,AhnGM12Linear,KiverisLMRV14, RastogiMCS13,Cohen2009graph, Kang2009pegasus}, 
and in the closely related distributed model of Congested Clique~\cite{Jurdzinski018,GhaffariP16,HegemanPPSS15,LotkerPPP03} (see, e.g.,~\cite{BehnezhadDH18} for the formal connection 
between the two models). In particular, for the \emph{sparse} connectivity problem,~\cite{KarloffSV10,KiverisLMRV14, RastogiMCS13} devised algorithms that achieve $O(\log{n})$ rounds using $n^{\Omega(1)}$ memory per machine and $O(n)$ total memory. 
In the classical PRAM model, $O(\log{n})$-round algorithms have been known for connectivity for over three decades 
now~\cite{Shiloach82logn, Reif85optimal, Gazit86optimal, KNP92,HalperinZ94}. 

In the truly sublinear regime of $n^{1-\Omega(1)}$ memory per-machine, $o(\log n)$-round MPC algorithm are only known for special cases. 
In \cite{AndoniNOY14}, Andoni et al. developed an approximate algorithms for  approximating minimum spanning tree and Earth Mover distance for geometric graphs (complete weighted 
graphs for points in geometric space). In \cite{Fischer18breaking}, Fischer and Uitto presented an $O((\log \log n)^2)$ rounds MPC algorithm for the maximal independent set problem (MIS) 
on trees. 

We point out that the MPC model is a special case of the Bulk-Synchronous-Parallel (BSP) model \cite{ValiantBSP90}, but has the  advantage of having fewer parameters.  
This makes algorithm design more ``coarse-grained" and streamlines the search for efficient algorithms, as evident by the omnipresence of this model in practice.

We refer the interested reader to~\cite{CzumajLMMOS17} and~\cite{LattanziFMS11} and references therein for further details on MPC algorithms for other graph problems. 

\subsection{Recent Development}\label{sec:developement}

Independently and concurrently to our work, Andoni~\etal~\cite{AndoniSSWZ18} have also studied MPC algorithms for the sparse connectivity problem with the goal of achieving 
improved performance on graphs with ``better connectivity'' structure by parametrizing based on the diameter of each connected component (as opposed to spectral gap in our paper). 
They develop an algorithm with $n^{\Omega(1)}$ memory per machine and $O(\log{D} \cdot \log\log_{N/n}{(n)})$ rounds, where $D$ is the largest diameter of any
connected component and $N = \Omega(m)$ is the total memory. Our results and that of~\cite{AndoniSSWZ18} are \emph{incomprable}: while in any graph $D = O(\log{n}/\lambda)$, the dependence 
on the number of rounds in~\cite{AndoniSSWZ18} is $O(\log{D} \log\log{n})$ for the main setting of interest in sparse connectivity when the total memory is within logarithmic factors of input size (the typical requirement of the MPC 
model\footnote{Minimizing total memory is \emph{critical} in the sparse connectivity problem in the MPC model. After all, the straightforward algorithm that computes the transitive closure of the graph (e.g. by matrix multiplication; see~\cite{KarloffSV10}) achieves 
$O(\log{D})$ rounds, subsuming both our results and~\cite{AndoniSSWZ18}; however, this algorithm requires \emph{at least} $\Omega(n^2)$ total memory and hence does not adhere to restrictions of MPC model (or any of its more relaxed variants such 
as~\cite{KarloffSV10}).}~\cite{AndoniNOY14,BeameKS13}). As such, our algorithm achieves {quadratically} smaller round complexity when the spectral gap is large, i.e., is $\Omega(1)$ or even $\Omega(1/\polylog{(n)})$ (as in random graphs and graphs with 
moderate expansion), while~\cite{AndoniSSWZ18} achieve better performance on graphs with small spectral gap but not too-large diameter (an example is two disjoint expanders connected by an edge). 

Both results at their core employ a leader-election algorithm for connectivity but the similarity between the techniques ends here. Our algorithm uses a random walk data structure (new to our paper) to transform each connected component of the input graph
to a \emph{random} graph, and \emph{after} that applies a novel leader-election algorithm to find components in $O(\log\log{n})$ rounds. On the other hand,~\cite{AndoniSSWZ18} design a leader-election algorithm
that runs in $O(\log\log{n})$ phases that are \emph{interleaved} with an $O(\log{D})$-round procedure that increases the degree of vertices in the remaining graph (by partially computing the transitive closure of the graph) to prepare for the next phase. We note 
that the combination of our random walk primitive and our leader-election algorithm for random graphs is the main reason we achieve the improved round complexity compared 
to~\cite{AndoniSSWZ18}, albeit by depending on spectral gap instead of diameter (the result of~\cite{AndoniSSWZ18} implies $\Omega((\log\log{n})^2)$ rounds even on our final random graph instances as diameter of these graphs is $\Omega(\log{n})$). 
We point out that our Theorem~\ref{thm:sublinear} is orthogonal to the results in~\cite{AndoniSSWZ18}.

\section{Preliminaries}\label{sec:prelim}

\paragraph{Notation.} For a graph $G(V,E)$, we define $V(G) = V$ and $E(G) = E$, and let $n = \card{V(G)}$ and $m = \card{E(G)}$. 
For any vertex $v \in V$, we use $d_v$ to denote the degree of $v$ in $G$. Throughout the paper, we assume without loss of generality that $d_v \geq 1$ for all vertices $v \in G$ (i.e., $G$ does not have isolated vertices).

For a graph $G(V,E)$, we say that a subset $C \subseteq V(G)$ is a \emph{component} of $G$ if the induced subgraph of $G$ on $C$ is connected. 
We say that a partition $\CC = \set{C_1,\ldots,C_k}$ of $V(G)$ is a component-partition iff every $C_i$ is a component of $G$. 

We denote the total variation distance between two distributions $\mu$ and $\nu$ on the same 
support as $\tvd{\mu}{\nu}$. We use the following basic property of total variation distance. 

\begin{fact}\label{fact:tvd-small}
	Suppose $\mu$ and $\nu$ are two distributions for an event $\event$, then, $\Pr_{\mu}(\event) \leq \Pr_{\nu}(\event) + \tvd{\mu}{\nu}$. 
\end{fact}
\noindent
A summary of concentration bounds used in this paper is presented in Appendix~\ref{app:concentration}.

\paragraph{Concise range notation.} For simplicity of exposition, we follow~\cite{CzumajLMMOS17} in using the following concise notation for representing ranges: for a value $x$ and parameter $\delta \geq 0$, we use
$\range{x \pm \delta}$ to denote the range $[x-\delta,x+\delta]$. We extend this notation to numerical expressions as follows: let $E$ be a numerical expression that apart from standard
operations also contains one or more applications of the binary operator $\pm$. Let $E_{+}$ be the expression obtained from $E$ by choosing assignment of $-$ and $+$ to replace different choices of the operator $\pm$ in order to maximize $E$; similarly, 
define $E_{-}$ for minimizing $E$. We now define $\range{E} := [E_-,E_+]$. 
For example, $\range{(3\pm 2)^2} = [1,25]$ and $\range{(2\pm 1)/(4 \pm 2)} = [1/6,3/2]$. 

\paragraph{Almost regular graphs.} Let $\Delta \geq 1$ be an integer and $\eps > 0$ be any parameter. We say that a graph $G(V,E)$ is $\range{\paren{1 \pm \eps} \Delta}$-almost-regular iff degree of any vertex in $V(G)$ belongs to 
$\range{\paren{1\pm \eps}\Delta}$. We refer to $\eps$ as the \emph{discrepancy factor} of the an almost-regular graph. 

\paragraph{Sort and search in the MPC model.} In MPC implementation of our algorithms, we crucially use
the by-now standard primitive of parallel sort and search introduced originally by~\cite{GoodrichSZ11}. 
On machines with memory $s$, the sort operation of~\cite{GoodrichSZ11} allows us to sort a set of $n$ key-value pairs in $O(\log_s{n})$ MPC rounds. We can also do a parallel search:
given a set $A$ of key-value pairs and a set of queries each containing a key of an element in $A$, we can annotate each query with the corresponding key-value pair from $A$, again in $O(\log_s{n})$ MPC rounds.

\subsection{Spectral Gap}\label{sec:spectral-gap}

Let $G(V,E)$ be an undirected graph on $n$ vertices. We use $A_{n \times n}$ to denote the adjacency matrix of $G$ and $D_{n \times n}$ to denote the diagonal matrix of degrees of vertices in $G$. 
We further denote the \emph{normalized Laplacian} of $G$ by $\LL := I - (D^{-1/2}\cdot A \cdot D^{-1/2})$. $\LL$ is a symmetric matrix with $n$ real eigenvalues $0=\lambda_1 \leq \lambda_2 \leq \ldots \leq \lambda_n \leq 1$. Throughout the paper, we use 
$\lambda_i(G)$ to refer to the $i$-th smallest eigenvalue $\lambda_i$ of normalized Laplacian $\LL$ of $G$. 

The quantity $\lambda_2(G)$ is referred to as the \emph{spectral gap} of $G$, and is a quantitative measure of how ``well-connected" the graph $G$ is. 
For example, it is well-known that $\lambda_2(G) > 0$ iff $G$ is connected (see, e.g.,~\cite{Chung97} for a proof), and the larger $\lambda_2(G)$ is, the 
graph $G$ is more ``well-connected'' under various notions of connectedness. For instance, cliques and expanders, two of the most well-connected graphs, 
have large spectral gap; see Cheeger's inequality~\cite{Cheeger70} for another such connection. In this paper, we also use $\lambda_2(G)$ as a measure 
of connectivity of $G$ and design algorithms with improved performance guarantee for graphs with larger spectral gap.

\subsection{Random Walk on Graphs}\label{sec:random-walk}

Let $G(V,E)$ be an undirected graph. Consider the random process that starts from some vertex $v \in V$, and repeatedly moves to a neighbor of the current vertex chosen uniformly at random. 
We refer to this process as a \emph{random walk}. In particular, a random walk of length $t$ corresponds to $t$ step of the above process. We refer to the distribution of the vertex reached by a random walk of length $t$ from
a vertex of $v$, as the distribution of this random walk and denote it by $\distrw(v,t)$. 

Define the \emph{random walk matrix} $W := D^{-1} \cdot A$. For any vector $v_i \in V$, let $\bm{e}_v$ denote the $n$-dimensional vector which is zero in all coordinates except for the $i$-th coordinate which is one. It is 
easy to see that for any integer $t \geq 1$, the vector $W^{t} \cdot \be_v$ corresponds to the distribution of a random walk of length $t$ starting from $v$, i.e., $\distrw(v,t)$. 
We use $\pi = \pi(G)$ to denote the \emph{stationary distribution} of a random walk on a graph $G$, where for any $v \in V$, $\pi_v := \frac{d_v}{2m}$. It is immediate to verify that $W \cdot \pi = \pi$. 

As random walks on arbitrary connected graphs do not necessarily converge to their stationary distribution (i.e., when the underlying graph is bipartite), we further consider \emph{lazy random walks}. In a lazy random walk of length $t$, 
starting from some vertex $v \in V$, for $t$ steps we either stay at the current vertex with probability half, or move to a neighbor of the current vertex chosen uniformly at random. We define the \emph{lazy random walk matrix} as 
$\bW := (I + W)/2$ which is the transition matrix of a lazy random walk. It is easy to verify that $\pi$ is also the stationary distribution for a lazy random walk. 

\paragraph{Mixing Time.} For any $\gamma > 0$, we define the $\gamma$-mixing time of $G$, denoted by $T_{\gamma}(G)$ to be the smallest integer $t \geq 1$, such that the distribution of a 
lazy random walk of length $t$ on $G$ starting from any arbitrary vertex become $\gamma$-close to the stationary distribution in total variation distance. Formally, 
\begin{align*}
	T_\gamma(G) := \min_{t \geq 1} \max_{v \in V(G)} \set{\tvd{\bW^t \cdot \be_v}{\pi} \leq \gamma}. 
\end{align*}


\noindent
The following well-known proposition relates the mixing time of a graph $G$ to its spectral gap (see,~e.g.~\cite{Chung97} chapter 1.5 for a proof). 

\begin{proposition}\label{prop:mixing-spectral}
Let $G(V,E)$ be any connected undirected graph. For any $\gamma < 1$, 
\begin{align*}
T_\gamma(G) = O\Paren{\frac{\log{(n/\gamma)}}{\lambda_2(G)}}.
\end{align*}
\end{proposition}

\subsection{Random Graphs}\label{sec:random-graph}

For any integers $n,d \geq 1$, we use $\PG(n,d)$ to denote the distribution on random undirected graphs $G$ on $n$ vertices chosen by picking for each vertex $v \in V(G)$, $\floor{d/2}$ 
outgoing edges $(u,v)$ for $v$ chosen uniformly at random (with replacement) from $V(G)$ and then removing the direction of edges. Note that this notion of a random graph is related but 
not identical to the more familiar family of Erdos-Renyi random graphs. 

Throughout the paper we use several properties of these random graphs that we present in this section. The proofs of the following propositions are standard and follow from similar arguments in Erdos-Renyi random graphs (we refer the interested
reader to~\cite{Bollobas98} for more details). For completeness, we provide simple proofs of these propositions in Appendix~\ref{app:random-graph}.

\begin{proposition}[\textbf{Almost-regularity}]\label{prop:PG-regular}
	Suppose $d \geq 4\log{n}/\eps^2$ for some parameter $\eps \in (0,1)$. A graph $G \sim \PG(n,d)$ is an $\range{\paren{1 \pm \eps}d}$-almost-regular with high probability. 
\end{proposition}
 
\begin{proposition}[\textbf{Connectivity}]\label{prop:PG-connected}
	A graph $G \sim \PG(n,d)$ for $d \geq c \log{n}$ is connected with probability at least $1-1/n^{(c/4)}$. 
\end{proposition}
	
\begin{proposition}[\textbf{Expansion}]\label{prop:PG-expansion}
	Suppose $G \sim \PG(n,d)$ for $d \geq c \log{n}$. Then, with probability at least $1-1/n^{(c/4)}$: 
	\vspace{-5pt}
	\begin{enumerate}
		\item For any set $S \subseteq V(G)$, the neighborset $N(S)$ of $S$ in $G$ has size $\card{N(S)} \geq \min\set{2n/3,d/12 \cdot \card{S}}$.
		\item The mixing time of $G$ is $T_{\gamma}(G) = O(d^2\cdot \log{(n/\gamma)})$ for any $\gamma < 1$.\footnote{The dependence on $d$ in the mixing time in Proposition~\ref{prop:PG-expansion} does not seem necessary; however, as we are working with the case when $d = O(\log{n})$ and our bounds depend only on $O(\log{d})$, we allow for this extra factor of $d^2$ which can greatly simplify the proof.}
	\end{enumerate}
\end{proposition}

\section{Technical Overview of Our Algorithm}\label{sec:technical}

In this section, we present a streamlined overview of our technical approach for proving Theorem~\ref{res:main}. For simplicity, we focus here mainly on the case  
$\lambda = 1/\polylog{(n)}$, i.e., the case of graphs with moderate (spectral) expansion. 

The general strategy behind our algorithm is the natural and familiar approach of \emph{improving the connectivity} of the underlying graph before finding its connected components (see, e.g., the celebrated log-space connectivity
algorithm of Reingold~\cite{Reingold05}). In particular, we perform the following transformations on the original graph: 
\begin{enumerate}[leftmargin=10pt]
	
	\item[] \textbf{Step 1: Regularization.} We first transform the original graph $G$ into an $O(1)$-regular graph $G_1$ such that $(i)$ there is a
	one to one correspondence between connected components of $G_1$ and $G$, and $(ii)$ mixing time of every connected component of $G_1$ is still $\polylog{(n)}$ (using the fact that $\lambda = 1/\polylog{(n)}$ and Proposition~\ref{prop:mixing-spectral}). 
	
	\item[] \textbf{Step 2: Randomization.} Next, we transform every connected component of $G_1$ to a random graph  
	chosen from distribution of random graphs $\PG$ to obtain a graph $G_2$. 
	This transformation (w.h.p) preserves all connected components of $G_1$ and never merges two separate components of $G_1$ into $G_2$. 
	As it turns out, the structure of random graphs (\emph{beyond} their improved connectivity) makes them ``easy to solve'' for MPC algorithms (more on this below).     
	
	\item[] \textbf{Step 3: Connectivity on random graphs.} Finally, we design a novel algorithm for finding connected components of $G_2$ which are each a random graph sampled from $\PG$.  
	This algorithm can be seen as yet another transformation that reduces the diameter of each component to $O(1)$ and then solve the problem on a 
	low-diameter graph using a simple broadcasting strategy.  
\end{enumerate}

\noindent
We now elaborate more on each step of this algorithm.

\paragraph{\underline{Step 1: Regularization.}} 
The main steps of our algorithm heavily rely on the properties of \emph{regular} graphs, for several important reasons that will become evident shortly. 
Our first goal is then to ``regularize'' the input while preserving its connected components, its spectral gap
and the number of edges (Lemma~\ref{lem:first-step}). 
The standard procedure for regularizing a graph by adding self-loops to vertices (e.g.,~\cite{Reingold05}) is too lossy for our purpose 
as it can dramatically reduce the spectral gap\footnote{Unlike vertex and edge expansion, spectral expansion is \emph{not} a 
monotone property of edges of the graph.}.


We instead use an approach based on the so-called \emph{replacement product} (see,~e.g.,~\cite{ReingoldVW00}): The idea is to replace each vertex $v$ of the original graph with degree 
$d_v$, by a $\Delta$-regular expander on $d_v$ ``copies of $v$", and then connect these expanders across according to edges of $G$ to construct a $(\Delta+1)$-regular graph (see Section~\ref{sec:regularization} for details). It is known that this product (approximately) preserves the
spectral gap in the new graph,  
hence the mixing time of each component remains $\polylog{(n)}$ even after this transformation. 
Implementing this approach in the MPC model has its unique challenges as the degree of some vertices in the original graph can be as large 
as $\Omega(n)$ hence we need a \emph{parallel} procedure for constructing the expanders and performing the product, as no machine can do these tasks locally on its $n^{\Omega(1)}$-size memory (see Lemmas~\ref{lem:parallel_regular_graph} and~\ref{lem:parallel_replacement}).  

We point out that replacement products have been used extensively in the context of connectivity and expansion 
to reduce the degree of \emph{regular} graphs 
\cite{ReingoldVW00,Reingold05}, 
but to best of our knowledge, our (distributed) implementation of this technique for the regularization purpose itself, 
while preserving its spectral gap, is nontrivial (yet admittedly quite anticipated\footnote{This in fact requires us to extend the proof of expansion of replacement product to non-regular graphs as all existing proofs of this result that we are aware of are assuming 
original graph is regular~\cite{ReingoldVW00,ReingoldTV06,RozenmanV05,hoory2006expander,trevisan2013lecture}, while our sole purpose is to regularize the graph; see Section~\ref{sec:regularization} for details and Appendix~\ref{app:replacement-zigzag} for this proof.}). 
We believe this parallel regularization primitive itself will be a useful building block in future MPC graph algorithms.

\paragraph{\underline{Step 2: Randomization.}} The goal of the second step is, morally speaking, 
to replace each connected component of the regular graph $G_1$ with a purely \emph{random graph} sampled from distribution $\PG$ with degree 
$O(\log{n})$ on the same connected component 
(which will indeed be connected with high probability by Proposition~\ref{prop:PG-connected}); This is the content of Lemma~\ref{lem:second-step}. 

In order to achieve the desired transformation, we need to connect every vertex $v$ in $G_1$ to $O(\log{n})$ \emph{uniformly random} vertices \emph{in the same connected component as $v$}. 
The obvious challenge in this step is that the information about which vertices belong to the same connected component is decentralized and each machine only has a ``local" view 
of the graph. To this end, we perform $O(\log{n})$ \emph{lazy random walks} of length $T = \polylog{(n)}$ from \emph{every vertex} of the graph $G$, 
where $T$ is an upper bound on the mixing time of every connected component of $G_1$.  
This, together with the fact that $G_1$ is \emph{regular}, ensures that the target of each random walk is (essentially) a uniformly random vertex in the corresponding connected component of $G_1$.

The main contribution in this step is an efficient parallel construction of a distributed data-structure for performing and manipulating \emph{independent} random walks of length $T$ in 
a regular graph, with only $O(\log{T})$ 
MPC rounds; see Theorem~\ref{thm:random-walk}. This allows us to 
perform the above transformation in $O(\log{T}) = O(\log\log{n})$ MPC rounds. 
Standard ideas such as recursively computing random walks of certain length from every vertex in parallel and then ``stitching'' these walks 
together to double the length of each walk can be used to implement this step (see~\cite{AleliunasKLLR79,KNP92,HalperinZ94} for similar implementations in the PRAM model)\footnote{
This step is also similar-in-spirit to streaming and distributed implementations 
of random walks in~\cite{SarmaGP11,SarmaNPT13}, with the difference that MPC algorithms can leverage the ``all-to-all'' communication to achieve an \emph{exponential} 
speed up of $O(\log{T})$ rounds as opposed to $O(\sqrt{T})$ achieved by these works, which is known to be tight~\cite{NanongkaiSP11}.}. 
The main challenge however, which is crucial for sampling from distribution $\PG$, is that in all the aforementioned implementations, the random walks produced across vertices are 
\emph{not independent of each other} as different walks become correlated once they ``hit'' the same vertex (the remainder of the walk would become the same for both walks). 

A key observation that allows us to circumvent this difficulty is that in a \emph{regular} graph, no vertex can become a ``hub'' which many different random walks hit 
(contrast this with a star-graph where every random walk almost surely hits the center); this is one of the key reasons that we need to perform the regularization step first. 
As such, many of the walks computed in the above procedure are indeed independent of each other. We use this 
observation along with several additional ideas (e.g., having each vertex compute multiple random walks and assign them randomly to different length walks in the 
recursive procedure above) to implement this step. 


\paragraph{\underline{Step 3: Connectivity on random graphs.}} The final and main step of the proof is an algorithm for identifying all connected components of a graph which
are each sampled from $\PG$ in only $O(\log\log{n})$ MPC rounds (Lemma~\ref{lem:third-step}). 
The centerpiece of this step is a leader-election based algorithm for connectivity (similar to most algorithms for sparse connectivity in the MPC model,~e.g.,~\cite{KarloffSV10,KiverisLMRV14,RastogiMCS13}). 
A typical leader-election algorithm for connectivity would pick some set of ``leader vertices'' in each round, and let other non-leader vertices connect to some leader in their neighborhood. 
It then ``contracts'' every leader vertex and all non-leader vertices that choose this leader to connect, to form a component of the input graph. This way, 
components of the graph ``grow" in each round as information propagates through leaders, until all components of the graph are discovered. 
The \emph{rate of growth} of components in these algorithms is however typically only a \emph{constant} as in general, it is hard to find components of size beyond a constant 
in each round (consider for instance the case when the underlying graph is a cycle). Consequently, $\Omega(\log{n})$ rounds are necessary to find all connected components using these algorithms. 

Our algorithm achieves an exponential speedup in the number of rounds by crucially using the properties of the random graphs $\PG$ to contract components which are \emph{quadratically} larger after each round, i.e., it grows a component of 
size $x$ into a component of size $x^2$ in each round. 

The intuition behind the algorithm is as follows. Let $H\sim\PG(n,d)$. Since $H$ is essentially $d$-regular (Proposition~\ref{prop:PG-regular}), sampling each vertex as a leader with probability 
$\Theta(1/d)$, we expect each non-leader vertex to have a \emph{constant} number of leader neighbors, say exactly $1$ for simplicity. 
Since every vertex has $d$ neighbors, contracting every leader vertex along with all of its non-leader neighbors into a single ``mega-vertex" will form components of size $d$ with   
total degree (roughly) $d^2$ 
in the contracted graph (this follows from the randomness in the distribution $\PG$ as no single mega-vertex is likely to be the endpoint of more than one of these $d^2$ edges).  
As such, the resulting graph after contraction is an almost $d^2$-regular random graph on $n/d$ vertices. By continuing this process on the new graph, 
we can now pick each leader with probability $1/d^2$ instead and contract components of size $d^2$ (instead of $d$). Repeating this process $i$ steps creates components of size 
$d^{2^{i}}$ which implies that after $O(\log\log{n})$ iterations we would be done. We stress that this algorithm exploits the ``entropy" of the distribution $\PG$ crucially, and not just the connectivity 
properties, e.g., expansion, of $\PG$, hence it is not clear (and unlikely) that this algorithm can be made to work directly on expander graphs (i.e., without Step 2).  

The outline above oversimplifies many details. Let us briefly mention two.  
Contracting vertices in this process \emph{correlates} the edges of $\PG$, impeding a recursive application.  
We bypass this problem by partitioning the edges of the random graph into $O(\log\log{n})$ different batches and running the algorithm (and analysis)  
in each round of the computation using a ``fresh random seed" (batch). This breaks the dependency between the choices made by the algorithm in previous rounds, 
and the randomness of the underlying graph. Another subtle issue is in the fact that the graphs in this process start ``drifting'' from regular to almost-regular with larger and larger discrepancy factors, indeed \emph{exponentially larger} after each round. At some point, this discrepancy factor becomes so large that one cannot anymore continue the previous argument. Fortunately however, as we are only performing $O(\log\log{n})$ rounds of computation, this only happens when size of each component has become $n^{\Omega(1)}$. At this point, we can simply 
stop the algorithm and argue that the diameter of the contracted graph is only $O(1)$. This allows us to run a simple algorithm for computing a BFS tree in this graph in $O(1)$ rounds, by computing levels of the tree one round at a time

\section{Step 1: Regularization}\label{sec:regularization}

In this section, we show how to ``preprocess'' our graph in order to prepare it for the main steps of our algorithm. Roughly speaking, this step
takes the original graph $G$ and turn it into a regular-graph without increasing its mixing time by much. Formally,

\begin{mdframed}[hidealllines=false,backgroundcolor=gray!10,innertopmargin=0pt]
\begin{lemma}\label{lem:first-step}
	There exists an MPC algorithm that given any graph $G(V,E)$ computes another graph $H$ with the following properties with high probability:
	\begin{enumerate}
		\item $\card{V(H)} = 2m$ and $H$ is $\Delta$-regular for some absolute constant $\Delta = O(1)$. 
		\item There is a one-to-one correspondence between the connected  components of $G$ and $H$. 
		\item Let $H_i$ be a connected component of $H$ corresponding to the connected component $G_i$ of graph $G$. For any $\gamma < 1$, $T_{\gamma}(H_i) = O\Paren{\frac{\log{(n/\gamma)}}{\lambda_2(G_i)}}$.
	\end{enumerate}	
	For any $\delta > 0$, the algorithm can be implemented on $O(m^{1-\delta})$ machines each with $O(m^{\delta})$ memory and in $O(\frac{1}{\delta})$ MPC rounds. 
\end{lemma}
\end{mdframed}

To prove Lemma~\ref{lem:first-step}, we use an approach based on the standard replacement product described in the next section.

\subsection*{Replacement Product}\label{sec:replacement-zig-zag}

Let $G$ be a graph on $n$ vertices $v_1,\ldots,v_n$ with degree $d_v$ for $v \in V(G)$, and $\HH$ be a family of $n$ $d$-regular graphs $H_1,\ldots,H_n$ where $H_v$ is supported on $d_v$ vertices (we assume $d_v \geq d$ for all $v \in V(G)$).
We construct the \emph{replacement} product $G \replacement \HH$ as follows: 
\begin{itemize}
	\item Replace the vertex $v$ of $G$ with a copy of $H_v$ (henceforth referred to as a \emph{cloud}). For any $i \in H_v$, we use $(v,i)$ to refer to the $i$-th vertex of the cloud $H_v$. 
	
	\item Let $(u,v)$ be such that the $i$-th neighbor of $u$ is the $j$-th neighbor of $v$. Then there exists an edge between vertices $(u,i)$ and $(v,j)$ in $G \replacement \HH$. Additionally, for any $v \in V(G)$, if there exists an edge $(i,j) \in H_v$,
	then there exists an edge $((v,i),(v,j))$ in $G \replacement \HH$. 
\end{itemize}
It is easy to see that the replacement product $G \replacement \HH$ is a $(d+1)$-regular graph on $2m$ vertices where $m$ is the number of edges in $G$. The following proposition asserts that the spectral gap is preserved under replacement product.

\begin{proposition}[cf. \cite{ReingoldVW00,ReingoldTV06}]\label{prop:replacement-product}
	Suppose $\lambda_2(G) \geq \lambda_G$ and all $H_v \in \HH$ are $d$-regular with $\lambda_2(H_v) \geq \lambda_H$. Then, 
	$\lambda_2(G \replacement \HH) = \Omega\paren{d^{-1} \cdot \lambda_G \cdot \lambda^2_H}$
\end{proposition}

Proposition~\ref{prop:replacement-product} was first proved in~\cite{ReingoldVW00} when $G$ is also a $D$-regular graph and all copies in $\HH$ are the same $d$-regular graph on $D$ vertices (in fact, all proofs of this proposition that
we are aware of, e.g.,~\cite{ReingoldVW00,ReingoldTV06,RozenmanV05,hoory2006expander,trevisan2013lecture}, are for this case). 
However, for our application, we crucially need this proposition for non-regular graphs $G$ (after all, our ultimate goal
is to ``regularize'' the graph). Nevertheless, extending these proofs to the case of non-regular graph $G$ as stated in Proposition~\ref{prop:replacement-product} is not hard and we provide a proof following the approach in~\cite{RozenmanV05,ReingoldTV06}
in Appendix~\ref{app:replacement-zigzag} for completeness.

For our purpose, we only need Proposition~\ref{prop:replacement-product} when every graph in $\HH$ is a constant-degree regular expander. In this case, since $\lambda_H = \Omega(1)$ and $d = O(1)$, we 
obtain that the resulting graph $G \replacement \HH$ has spectral gap at least $\lambda_2(G \replacement \HH) = \Omega(\lambda_2(G))$. 
 
\subsection*{Parallel Expander Construction}

To use Proposition~\ref{prop:replacement-product}, we need to be able to create a family of expanders $\HH$ in parallel over the set of vertices of the original graph $G$. This is a non-trivial task as degree of some vertices in $G$ can be 
as high as $\Omega(n)$ and hence we need to create an expander with $\Omega(n)$ edges to replace them; at the same time, no single machine has $\Omega(n)$ memory to fit this expander and hence it should be constructed in parallel and distributed
across multiple machines. We note that however, we can use a randomized algorithm for this task (i.e., we do not need necessarily an ``explicit'' construction). 

Consider the following construction of a random $d$-regular undirected graph on $n$ vertices for positive even integer $d$ (allowing self-loops and parallel edges):
Let $\pi_1, \dots, \pi_{d/2}$ be $d/2$ permutations on $[n]$ which are independently and uniformly sampled from the set of all permutations.
The resulting graph is $H$ with $V(H) := [n]$ and 
\begin{equation}\label{eq:edge} E(H) := \set{(i, \pi_j(i)) : i \in [n], j \in [d / 2]}\end{equation}
for unordered pairs $(i, \pi_j(i))$. Let $\mathcal{G}_{n, d}$ be the probability space of  the $d$-regular $n$-vertex graphs constructed in this way.

\begin{proposition}[c.f.~\cite{Friedman03}]\label{prop:random_regular_graph}
Given a postive constant $\delta > 0$ and an positive even integer $d$,
there is a constant $c$ such that 
\[\Pr_{H \sim \mathcal{G}_{n, d}}\left[\lambda_2(H) \geq 1 -  \frac{2\sqrt{d - 1} + \delta}{d}\right] \geq 1 - \frac{c}{n^{\tau}},\]
where $\tau = \lceil (\sqrt{d - 1} + 1) / 2\rceil - 1$.
\end{proposition}

We choose $d$ to be 100. By Proposition~\ref{prop:random_regular_graph}, we have 

\begin{corollary}\label{cor:random_regular_graph}
Let $d = 100$. There is a constant $c$ such that for any positive integer $n$
\[\Pr_{H \sim \mathcal{G}_{n, d}}\left[\lambda_2(H) \geq \frac{4}{5}\right] \geq 1 - \frac{c}{n^{5}}.\]
\end{corollary}

In the remaining of this section, we present an MPC algorithm to construct random $d$-degree graphs for a given sequence of positive integers $n_1, n_2, \dots, n_k$ satisfying $\sum_{i=1}^k n_i \leq 2m$.

\begin{tbox}
	$\RGC(m^{\delta},n_1, \dots, n_k)$.  An algorithm for constructing random $d$-regular graphs with $n_1, n_2, \dots, n_k$ vertices  for $d = 100$ on machines of memory $O(m^{\delta})$. 
	
	\medskip
	
	\textbf{Output:} $d$-regular graphs $\bH_{n_i}$ for $1 \leq i \leq k$.   
	
	\algline
	
	\begin{enumerate}
			\item 
			For every $n_i \leq m^{\delta}$ \textbf{in parallel} repeat the following process until the resulting graph $\bH_{n_i}$ satisfies $\lambda_2(\bH_{n_i}) \geq 4/5$:
			uniformly sample $d/2$ permutations $\pi_1, p_2, \dots \pi_{d/2}$ on $[n_i]$, and construct graph $\bH_{n_i}$ by Eq.~(\ref{eq:edge}).
			
			\item 
			For every $n_i > m^\delta$ \textbf{in parallel} construct $\bH_{n_i}$ on $d \cdot \lceil n_i / m^\delta\rceil$ machines
			\begin{enumerate}
				\item Independently and uniformly sample $v_{n_i, j, k}$ from $[n^{10}]$ for all $j \in [n_i], k\in[d/2]$.
				\item For every $k \in [d/2]$, sort $\{ v_{n_i, 1, k}, \dots, v_{n_i, n_i, k} \} $, and set $\pi_{n_i, k}(j)$ to be $\alpha$
				if $v_{j, k}$ is $\alpha$-th largest number among $\{ v_{n_i, 1, k}, \dots, v_{n_i, n_i, k} \} $.
				\item Construct graph $\bH_{n_i}$ using $\pi_{n_i, 1}, \dots, \pi_{n_i, d/2}$ with edge set specified in Eq.~(\ref{eq:edge}).
			\end{enumerate}

	\end{enumerate}
\end{tbox}

We use $\RGC$ to prove the following lemma.
\begin{lemma}\label{lem:parallel_regular_graph}
There exists an MPC algorithm that given a sequence of positive integers $n_1, n_2, \dots, n_k$ satisfying $\sum_{i=1}^k n_i \leq 2m$, with high probability,
 computes a set of graphs $\bH_{n_1}, \bH_{n_2}, \dots, \bH_{n_k}$ such that for every $\bH_{n_i}$ for $1 \leq i \leq k$, $\bH_{n_i}$ is a $d$-degree regular graph with $\lambda_2(\bH_{n_i}) \geq 4/5$.

For any $\delta > 0$, the algorithm can be implemented with $O(m^{1-\delta})$ machines each with $O(m^{\delta})$ memory, and in $O(1/\delta)$ MPC rounds. 
\end{lemma}

\begin{proof}
We first show the correctness of the $\RGC$ algorithm.
By Proposition~\ref{prop:random_regular_graph},
step 1 construct regular graphs with desirable spectral gap with high probability for every $n_i \leq m^\delta$. 
Now we show that step 2 construct regular graphs with desirable spectral gap with high probability for every $n_i > m^\delta$ as well.
For every $n_i \geq m^\delta$ and $k \in [d/2]$,
the probablity that $v_{n_i, 1, k}, \dots, v_{n_i, n_i, k}$ are distinct is at least
\[1 \cdot 2\cdot \ldots  \left(1 - \frac{n_i - 1}{n^{10}}\right) > \left(1 - \frac{n_i}{n^{10}}\right)^{n_i} \geq 1 - \frac{n_i^2}{n^{10}} \geq 1 - \frac{1}{n^8}.\]
If $v_{n_i, 1, k}, \dots, v_{n_i, n_i, k}$ are distinct, then $\pi_{n_i, k}$ is a random permutation among all the permutations on $[n_i]$, since all the permutations are constructed with same probability.
Conditioned on this, $\bH_{n_i}$ is a graph sampled from $\mathcal{G}_{n, d}$.
By union bound, with probability $1 - \frac{1}{n^7}$, 
step 2(c) obtain $\bH_{n_i}\sim\mathcal{G}_{n_i, d}$ for every $n_i > m^\delta$.
By Corollary~\ref{cor:random_regular_graph}, if $\bH_{n_i} \sim \mathcal{G}_{n_i, d}$, then 
$\lambda_2(\bH_{n_i}) \geq 4/5$ with probability $1 - \frac{c}{n_i^5}$ for some constant $c \geq 0$.
By union bound, all the $\bH_{n_i}$ constructed satisfying $\lambda_2(\bH_{n_i}) \geq 4/5$
with probability at least 
\[1 - \sum_{\ell = n^\eps}^n\frac{c}{\ell^5} \geq 1 - O\left(\frac{\log n}{n^4}\right).\]
Hence, the algorithm gives us graphs with desirable spectral gap with probability at least $1 - \frac{1}{n^3}$.

In the implementation of step 1, we assign every $n_i \leq m^{\delta}$ to a single machine such that for every machine, 
the sum of $n_i$ assigned to it is at most $O(m^\delta)$. Hence, step 1 can be done in $O(1)$ MPC rounds.

In step 2, for each $n_i$ and $k \in [d/2]$, 
we use $\lceil n_i / m^\delta \rceil$ machines to sample $v_{n_i, j, k}$ for all the $j \in [n_i]$ in $O(1)$ MPC rounds.
Sorting $\{ v_{n_i, 1, k}, \dots, v_{n_i, n_i, k}\}$ can be done in $O(1/\delta)$ MPC rounds on the same machines (see Section~\ref{sec:prelim}). 
Then $\pi_{n_i, k}(j)$ and thus edges of $\bH_{n_i}$ can be computed locally after sort. This concludes the proof.
\end{proof}

\subsection*{Parallel Replacement Product}

We present an MPC implementation of replacement product $G \replacement \HH$, where $\HH = \{H_v : v\in V\}$ is defined as follows: 
For every $v \in V$,  $H_v$ is a copy of $\bH_{d_v}$, where $\bH_{d_v}$ are the $d$-degree regular graphs with $d_v$ vertices returned by $\RGC(m^{\delta},d_{v_1},\ldots,d_{v_n})$.

\begin{lemma}\label{lem:parallel_replacement}
Given a graph $G(V, E)$ and $\bH_{d_v}$ for every $v \in V$, 
there is an MPC algorithm to compute $G \replacement \HH$, where 
$\HH = \{H_v : v\in V\}$ such that $H_v$ is a copy of $\bH_{d_v}$ for every $v \in V$.

	The algorithm can be implemented with $O(m^{1-\delta})$ machines each with $O(m^{\delta})$ memory, and in $O(1/\delta)$ MPC rounds. 
\end{lemma}

Lemma~\ref{lem:parallel_replacement} is obtained by the definition of replacement product and the following algorithm. 
\begin{tbox}
	$\PR(G, \set{\bH_{d_v} \text{ for every } v\in V(G)})$.  An algorithm for constructing $G \replacement \HH$. 
	
	\medskip
	
	\textbf{Output:} $H:=G \replacement \HH$.   
	
	\algline
	
	\begin{enumerate}
			\item 
			For every $v \in V(G)$ \textbf{in parallel} set $H_v$ be a copy of $\bH_{d_v}$, and let $H$ be initially $\cup_{v \in V(G)} H_v$.

			\item 
			For every edge $(u, v) \in E$ \textbf{in parallel} 
			where $v$ is $i$-th neighbor of $u$ and $u$ is $j$-th neighbor of $v$ in $G$, 
			add an edge to $H$ between $i$-th vertex of $H_u$ and $j$-th vertex of $H_v$. 
			\item Return $H$
			
	\end{enumerate}
\end{tbox}
The proof of correctness of this algorithm is straightforward. 

\subsection*{Proof of Lemma~\ref{lem:first-step}}

By Lemma~\ref{lem:parallel_regular_graph} and Lemma~\ref{lem:parallel_replacement}, we can compute the replacement product $H:= G \replacement \HH$ where $\HH$ is a family of 
graphs such that for all $v \in V(G)$, $\lambda_2(H_v) \geq 4/5$. By definition of replacement product and since $d = O(1)$, we obtain that $\card{V(H)} = 2m$ and $H$ is $\Delta$-regular
for $\Delta = d+1 = O(1)$. This proves  the first part of the lemma. 

Consider any connected component $G_i$ of $G$ and define $\HH_i := \set{H_v \in \HH : v \in V(G_i)}$. It is immediate to see that the subgraph of $H$ induced on vertices of 
$V(G_i) \times V(H_v)$ for $v \in V(G_i)$ (informally speaking, the vertices added to $H$ because of $G_i$) is exactly $G_i \replacement \HH_i$ which we denote by $H_i$. As replacement product preserves connectivity, 
$H_i$ is a connected component of $H$, hence proving the second part of the lemma. 

Finally, as $H_i = G_i \replacement \HH_i$ and $\lambda_2(H_v) = \Omega(1)$ for all $H_v \in \HH_i$, by Proposition~\ref{prop:replacement-product}, $\lambda_2(H_i) = \Omega(\lambda_2(G_i))$ (recall that $d = O(1)$). 
As such, by Proposition~\ref{prop:mixing-spectral}, mixing time 
\[T_\gamma(H_i) = O\Paren{\frac{\log{(n/\gamma)}}{\lambda_2(H_i)}}= O\Paren{\frac{\log{(n/\gamma)}}{\lambda_2(G_i)}},\]
 for any $\gamma < 1$, concluding the proof of the third part.
 
Implementation details of the algorithm follow immediately from Lemmas~\ref{lem:parallel_regular_graph} and~\ref{lem:parallel_replacement}.
\Qed{Lemma~\ref{lem:first-step}}

\section{Step 2: Randomization}\label{sec:randomization}

We present the second step of our algorithm in this section. Roughly speaking, this step transforms each connected component of the graph
into a ``random graph'' (according to the definition of distribution $\PG$ in Section~\ref{sec:prelim}) on the same set of vertices. Formally,

\begin{mdframed}[hidealllines=false,backgroundcolor=gray!10,innertopmargin=0pt]
\begin{lemma}\label{lem:second-step}
	Suppose $G(V,E)$ is any $n$-vertex $\Delta$-regular graph such that $T_{\gamstar}(G_i) \leq T$ for $\gamstar := n^{-10}$ and for all connected component $G_i$ of $G$.
	There exists an MPC algorithm that given $G$ and integer $T$ computes another graph $H$ with the following properties with high probability:
	\begin{enumerate}
		\item $V(H) = V(G)$, $\card{E(H)} = O(n)$ and each connected component $G_i$ of $G$ corresponds to a connected component $H_i$ of $H$ on $V(H_i) = V(G_i)$. 
		\item The connected component $H_i$ of $H$ is a random graph on $n_i = \card{V(H_i)}$ vertices sampled from the distribution $\distribution{H_i}$ such that $\tvd{\distribution{H_i}}{\PG(n_i,100 \log{n})} = n^{-8}$. 
	\end{enumerate}	
	For any $\delta > 0$, the algorithm can be implemented with $O(T^2 \cdot n^{1-\delta} \cdot \Delta \log^2{n})$ machines each with $O(n^{\delta})$ memory and in $O(\frac{1}{\delta} \cdot \log{T})$ MPC rounds. 
\end{lemma}
\end{mdframed}

We point out that the choice of constant $100$ in $\PG(n_i,100 \cdot \log{n})$ in Lemma~\ref{lem:second-step} is arbitrary and any sufficiently large constant (say larger than $8$) suffices for our purpose (similarly also for $\gamma^*$). 

To prove Lemma~\ref{lem:second-step}, we design a general algorithm for performing \emph{independent} random walks in the MPC model which can be of independent of interest. 
Let $G(V,E)$ be a $\Delta$-regular graph and $W = \Delta^{-1} \cdot A$ be its random walk matrix (note that this is scalar product with $\Delta^{-1}$ as $G$ is $\Delta$-regular). 
For any vertex $u \in V$, and integer $t \geq 1$, the vector $W^t \cdot \bm{e}_u$ denotes the distribution 
of a random walk of length $t$ starting from $u$ where $\be_u$ is an $n$-dimensional vector which is all zero except for the entry $u$ which is one. We use $\distrw(u,t) = W^{t} \cdot \bm{e}_u$ to denote this distribution.

\begin{theorem}\label{thm:random-walk}
	There exists an MPC algorithm that given any $\Delta$-regular graph $G(V,E)$ and integer $t \geq 1$, outputs a vector $(v_1,\ldots,v_n)$ such that with high probability: 
	\begin{enumerate}
	\item For any $i \in [n]$, $v_i$ is sampled from $\distrw(u_i,t)$, where $u_i$ is the $i$-th vertex in $V$. 
	\item The choice of $v_i$ is \emph{independent} of all other vertices $v_j$. In other words, $(v_1,\ldots,v_n)$ is sampled from the product distribution $\bigotimes_{i=1}^{n} \distrw(u_i,t)$
	\end{enumerate}
	For any $\delta > 0$, the algorithm can be implemented with $O(t^2 \cdot n^{1-\delta} \cdot \Delta\log{n})$ machines each with $O(n^{\delta})$ memory and in $O(\frac{1}{\delta} \cdot \log{t})$ MPC rounds. 
\end{theorem}

\subsection{Proof of Theorem~\ref{thm:random-walk}: The Random Walk Algorithm}
We start by presenting a parallel algorithm for proving Theorem~\ref{thm:random-walk} without getting into the exact details of its implementation, and then present an MPC implementation of this parallel algorithm. 
We start by introducing a key data structure in our algorithm. 

\subsection*{Layered Graph}\label{sec:layered-graphs}
A key component of our algorithm in Theorem~\ref{thm:random-walk} is the notion of a \emph{layered graph} which we define in this section and present its main properties. 

\begin{definition}[Layered Graph]
	For a graph $G(V,E)$ and integer $t \geq 1$, the layered graph $\FG(G,t)$ of $G$ is defined as the following \emph{directed} graph: 
	\begin{enumerate}
		\item \textbf{Vertex-set:} The vertex-set $\FV$ of $\FG$ is the set of all triples $(u,i,j) \in V \times [2t] \times [t+1]$. 
		\item \textbf{Edge-set:} There is a directed edge $(u,i,j) \rightarrow (v,\ell,k)$ in $\FG$ whenever $(u,v) \in E$ and $k = j+1$ for all choice of $i$ and $\ell$. 
	\end{enumerate}
	Throughout the paper, we use greek letters to denote the vertices in the layered graph.
\end{definition}

\noindent
 For any vertex $\alpha = (u,i,j) \in \FV$, we define $v(\alpha) = u \in V$. 
We partition the set of vertices $\FV$ into $t+1$ sets $\FV_1,\ldots,\FV_{t+1}$ where the $j$-th set consists of all vertices $(u,i,j)$ for $u \in V$ and $i \in [2t]$. We refer to each set $\FV_j$ as a \emph{layer} of the graph $\FG$. 
It is immediate to see that $\FG$ consists of $t+1$ layers and all edges in $\FG$ are going from one layer to the next.  Additionally, any vertex $u \in V$, contains $2t$ ``copies'' in every layer. As such, any edge in $E$ is mapped to $t$ directed bi-cliques on
the $2\cdot 2t$ copies of its endpoints between every two consecutive layers of $\FG$. 

\emph{\underline{Paths and walks in $\FG$ and $G$:}} The main property of the layered graph $\FG$ that we use is that any path starting from $\FV_1$ and ending in $\FV_{t+1}$ in $\FG$ corresponds to a walk of length $t$ (but not necessarily a path) in $G$. 
More formally, consider a path $\PP_\alpha = \alpha_1,\alpha_2,\ldots,\alpha_{t+1}$ where $\alpha = \alpha_1$ belongs to  $\FV_1$. We  can associate to $\PP_{\alpha}$ a walk
of length $t$ in $G$ starting from the vertex $v = v(\alpha)$, denoted by $W(\PP_{\alpha})$, in a straightforward way by traversing the vertices $u_i = v(\alpha_i)$ for $\alpha_i \in \PP_{\alpha}$.

\paragraph{Sampled layered graph.} In our algorithm, we work with a random subgraph of the layered graph defined as follows:  
For any vertex in $\FG$ \emph{independently}, we pick {exactly} one of its \emph{outgoing} edges \emph{uniformly at random} to form a subgraph $\FGS$, referred to as the \emph{sampled layered graph}. 

As the out-degree of any vertex in $\FGS$ is exactly one, starting from any vertex $\alpha \in \FV_1$, there is a \emph{unique} path $\PP_{\alpha}$ of length $t$ in $\FGS$ from $\alpha$ to some vertex $\beta \in \FV_{t+1}$.
It is easy see that a path $\PP_\alpha$ in $\FGS$ corresponds to a \emph{random} walk of length $t$ in $\FG$ starting from the vertex $v(\alpha)$ and ending in $v(\beta)$ (the randomness comes from the choice of $\FGS$).
We have
the following key observation. 

\begin{observation}\label{obs:multi-ind-layered-graph}
	Suppose $\PP_{\alpha_1},\ldots,\PP_{a_k}$ are $k$ \emph{vertex disjoint paths} from $\FV_1$ to $\FV_{t+1}$ in $\FGS$.  Then, the associated walks $W(\PP_{\alpha_1}),\ldots,W(\PP_{\alpha_k})$ 
	form $k$ \emph{independent} random walks of length $t$ in $G$. 
\end{observation}

\noindent
The justification for Observation~\ref{obs:multi-ind-layered-graph} is the simple fact that vertex disjoint paths in $\FGS$ do not share any randomness in choice of their neighbors.

In the rest of this section, we show that a sampled layered graph contains $\Omega(n)$ vertex disjoint paths from the first layer to the last one with high probability. Intuitively, this allows us to ``extract'' $\Omega(n)$ independent random walks 
from a sampled layered graph. We then use this fact in the next section to design our algorithm for simulating independent random walks in $G$.

Define $\FV^*_1 \subseteq \FV_1$ as the set of all vertices $(v,1,1) \in \FV_1$ for $v \in V$. We prove that the $\Omega(n)$ vertex disjoint paths mentioned above can all be starting from vertices in $\FV^*_1$. Formally,
\begin{lemma}\label{lem:independent-paths}
	For any vertex $\alpha \in \FV^*_1$, $\PP_{\alpha}$ in $\FGS$ is vertex disjoint from $\PP_{\beta}$ for all $\beta\neq \alpha \in \FV^*_1$, with probability at least $1/2$. 
\end{lemma}
\noindent
We emphasize that in Lemma~\ref{lem:independent-paths}, the paths starting from $\FV^*_1$ are \emph{only} guaranteed to be vertex disjoint with constant probability from other paths starting from $\FV^*_1$ and not all of $\FV_1$.  Before getting into
the proof of Lemma~\ref{lem:independent-paths}, we prove the following auxiliary claim regarding the number of paths of certain lengths in $\FG$ (not in $\FGS$). 

\begin{claim}\label{clm:paths-FG}
	For any layer $j \in [t+1]$ and any vertex $\alpha \in \FV_j$, the number of paths in $\FG$ that start from some vertex in $\FV^*_1$ and end in vertex $\alpha$ is $P_{j} = \paren{\Delta^{j-1} \cdot (2t)^{j-2}}$. 
\end{claim}
\begin{proof}
	Let $\alpha = (v,i,j)$ be in layer $j$. Since $G$ is $\Delta$-regular, $v \in V$ has exactly $\Delta$ neighbors in $V$. By construction of $\FG$, this means that $v$ has $\Delta \cdot (2t)$ neighbors in $\FV_{j-1}$ and hence there
	 are $\Delta \cdot (2t)$ paths of length $1$ that end up in $\alpha$.
	Similarly, the starting point of any of these paths has exactly $\Delta \cdot (2t)$ neighbors in $\FV_{j-2}$ and hence there are $(\Delta \cdot 2t)^2$ paths of length $2$ that can end up in $\alpha$. 
	Continuing this inductively, we obtain that there are $\paren{\Delta \cdot (2t)}^{j-1}$ paths of length $j$ that can reach the vertex $\alpha$. By the layered structure of the graph $\FG$, it is clear that all these paths need to start from a vertex in $\FV_1$.
	
	Furthermore, if $(u,i,1)$ (for some $u \in V$ and $i \in [2t]$) is starting point one of these paths, then for all $\ell \in [2t]$, $(u,\ell,1)$ would also be a starting point of one such path (this is because neighborset of all vertices $(u,\ell,1)$ is the same). 
	As such, exactly $1/(2t)$ fraction of these starting points belong to $\FV^*_1$ and hence there are $P_{j}:= \paren{\Delta^{j-1} \cdot (2t)^{j-2}}$ paths in $\FG$ that start from a vertex in $\FV^*_1$ and end in vertex $\alpha$.
	\Qed{Claim~\ref{clm:paths-FG}}

\end{proof}

\begin{proof}[Proof of Lemma~\ref{lem:independent-paths}]
	Let $\PP_\alpha = \alpha_1,\alpha_2,\ldots,\alpha_{t+1}$ where $\alpha_1 = \alpha$ and each $\alpha_j$ belongs to $\FV_j$ for $j > 1$. We define the following $t+1$ random variables $X_1,\ldots,X_{t+1}$, where $X_j$ counts the number
	of paths that start from a vertex $\beta \neq \alpha \in \FV^*_1$ and contain vertex $\alpha_j$ (as their $j$-th vertex). In other words, $X_j$ counts the number of paths that ``hit'' $\PP_\alpha$ in layer $\FV_j$. 
	
	Clearly, $X_1 = 0$. For any $j > 1$, we further define indicator random
	variables $Y_{j,1},Y_{j,2},\ldots,Y_{j,P_j}$ where $P_j$ (the quantity bounded in Claim~\ref{clm:paths-FG}) 
	is the number of paths that start from $\FV^*_1$ and end in $\alpha_j$ in $\FG$: for all $i \in [P_j]$, $Y_{j,i} = 1$ iff the $i$-th path (according to any arbitrary
	ordering) is fully appearing in $\FGS$ as well. Clearly, $X_{j} = \sum_{i} Y_{j,i}$. Hence, by linearity of expectation,	
	\begin{align}
		\Ex\bracket{X_j} &= \sum_{i=1}^{P_j} \Pr\paren{Y_{j,i} = 1} = \card{P_j} \cdot \paren{\frac{1}{\Delta \cdot (2t)}}^{j-1} \Eq{Claim~\ref{clm:paths-FG}} \frac{1}{2t}. \label{eq:X_j-expectation}
	\end{align}
	The second equality above is because in $\FGS$, each edge in the path has probability of $\frac{1}{\Delta \cdot (2t)}$ to appear (as out-degree of any vertex in $\FG$ is $\Delta \cdot (2t)$ and we are picking one of these edges uniformly at random
	in $\FGS$; moreover, the edges of a path appear independently in $\FGS$).
	
	Finally, notice that $X := \sum_{j=1}^{t+1} X_j$ counts the total number of paths starting from vertices in $\FV^*_1$ that can ever ``hit'' $\PP_\alpha$ in any layer.  Hence, $\Ex\bracket{X} = {1}/{2}$ by Eq~(\ref{eq:X_j-expectation}) (recall that $X_1 = 0$) and by Markov bound, 
	$\Pr\paren{X = 1} \leq {1}/{2}$. This implies that with probability at least $1/2$, $\PP_\alpha$ is vertex disjoint from any other path starting from a vertex in $\FV^*_1$. 
	\Qed{Lemma~\ref{lem:independent-paths}}
	
\end{proof}

\subsection*{A Parallel Random Walk Algorithm}

We now present a parallel algorithm for performing independent random walks of fixed length from every vertex of the graph. We start by presenting an algorithm with a weaker guarantee: in this algorithm only $\Omega(n)$ vertices are able to achieve a \emph{truly} independent random walk destination; moreover, these vertices are unknown to the algorithm. 
We then present a subroutine for detecting these $\Omega(n)$ vertices. Finally, we combine these two subroutines to obtain our final algorithm. 

Recall that for any vertex $u \in V(G)$ and integer $t \geq 1$, $\distrw(u,t)$ is the distribution of a random walk of length $t$ from $u$. We present the following algorithm. 

\begin{tbox}
	$\SRW(G,t)$.  An algorithm for performing a random walk of length $t$ from every vertex in a given graph $G$. 
	
	\medskip
	
	\textbf{Output:} For any vertex $u_i \in V(G)$, a vertex $v_i \in V(G)$ such that $v_i \sim \distrw(u_i,t)$.   
	
	\algline
	
	\begin{enumerate}
			\item Randomly sample a sampled subgraph $\FGS$ from the layered graph $\FG(G,t)$.
			\begin{enumerate}
				\item Set $\FV = V(G) \times [2t] \times [t+1]$, and distribute the vertices of $\FV$  to all the machines such that 
				each machine contains $O(n^\delta)$ vertices. 
				\item For every vertex $\alpha = (v, i, j) \in \FV$ such that $j \leq t$ \textbf{in parallel} independently and uniformly sample a number $n_\alpha$ from $[\Delta]$ and $i_\alpha$ from $[2t]$.
				\item Set $\FGS$ to be empty initially.
				\item For every vertex $\alpha  = (v, i, j) \in \FV$ such that $j \leq t$  \textbf{in parallel} set $v_\alpha$ to be $n_\alpha$-th neighbor of $v$ in $G$, and add an edge from $\alpha$ to $(v_\alpha, i_\alpha, j+1)$ to $\FGS$.
			\end{enumerate}
			\item For any vertex $\alpha \in \FGS$, define $N_0(\alpha) = \beta$ where $(\alpha,\beta) \in \FGS$ is the (only) outgoing edge of $\alpha$ in $\FGS$ (define $\beta = \perp$ if $\alpha$ belongs to $\FV_{t+1}$ and hence has no outgoing edge). 
			\item\label{for:srw} For $i=1$ to $\log{t}$ phases: For every $\alpha \in \FGS$ \textbf{in parallel} let $N_{i}(\alpha) = N_{i-1}(N_{i-1}(\alpha))$ (assuming $N_{i-1}(\perp) = \perp$).
			\item For any $\alpha \in \FV^*_1$, return $v = v(N_{\log{t}}(\alpha))$ as the target of the vertex $u = v(\alpha)$ (recall that $\FV^*_1$ is the set of all vertices $(u,1,1) \in \FV$ for $u \in V(G)$).
	\end{enumerate}
\end{tbox}

We first have the following simple claim. 

\begin{claim}\label{clm:srw-target}
	For any vertex $\alpha \in \FV^*_1$ of $\FGS$, $N_{\log{t}}(\alpha)$ is the endpoint of the path $\PP_\alpha$ in $\FGS$.  
\end{claim}
\begin{proof}
	We prove by induction that $N_{i}(\alpha)$ is the vertex at distance $2^{i}$ from $\alpha$ in $\PP_\alpha$. The base case for $i=0$ is true as $N_0(\alpha) = \beta$ where $\beta$ is the endpoint of the outgoing edge of $\alpha$. 
	For $i > 0$, by induction, $N_{i-1}(\alpha)$ is the vertex $\theta$ at distance $2^{i-1}$ from $\alpha$ and $N_{i-1}(\theta)$ is the vertex at distance $2^{i-1}$ from $\theta$. 
	Hence $N_i(\alpha) = N_{i-1}(N_{i-1}(\alpha))$ is at distance $2^{i}$ from $\alpha$ (as $\FGS$ is a directed acyclic graph with edges going only from one layer to the next). 
	As such, $N_{\log{t}}(\alpha)$ is at distance $t$ from $\alpha$ and hence is the endpoint of the path $\PP_{\alpha}$. 
\end{proof}

We say that $\SRW(G,t)$ \emph{finds} the vertex $v$ for $u$ if $v$ is returned as the target vertex of $u$. 
Claim~\ref{clm:srw-target} combined with Observation~\ref{obs:multi-ind-layered-graph} already implies that for any vertex $u \in V(G)$, the vertex $v$ found by $\SRW$ is distributed according to $\distrw(u,t)$.  
We further have,
\begin{lemma}\label{lem:srw}
	For any vertex $u \in G$, $\SRW(G,t)$ finds a vertex $v \sim \distrw(u,t)$ such that with probability at least $1/2$, $v$ is \emph{independent} of all other vertices found by $\SRW$. 
\end{lemma}
\begin{proof}
	Follows immediately from Claim~\ref{clm:srw-target}, Observation~\ref{obs:multi-ind-layered-graph}, and Lemma~\ref{lem:independent-paths}. 
\end{proof}

By Lemma~\ref{lem:srw}, we are able to find $\Omega(n)$ independent random walks in $G$ with high probability. However, a-priori it is not obvious how to detect these walks. In the following, we briefly describe 
a simple parallel procedure for this task. 

\paragraph{Detecting independent random walks.} The idea is to first find the path $\PP_{\alpha}$ for every $\alpha \in \FV^*_1$ and then remove any $v(\alpha)$ from consideration if $\PP_{\alpha}$ intersects another path
\emph{starting from} $\FV^*_1$. To do this, we need the following recursive ``marking'' procedure for marking all vertices on a path $\PP_{\alpha}$:

\begin{tbox}
${\Mark(\alpha,\beta,k):}$ An algorithm for marking all vertices in the path $\PP_{\alpha}$ recursively. 

\algline

			\begin{enumerate}
				\item Mark the vertex $\beta \in \FV$ with label $\alpha$. 
				\item If $k = 0$ stop. Otherwise recurse on $\Mark(\alpha,\beta,k-1)$ and $\Mark(\alpha,N_{k-1}(\beta),k-1)$.
			\end{enumerate}
\end{tbox}		
It is easy to see that by running $\Mark(\alpha,\alpha,\log{t})$ for every $\alpha \in \FV^*_1$ we can mark all vertices across all paths $\PP_{\alpha}$ (this can be proven inductively using an argument similar to Claim~\ref{clm:srw-target}). 
We remove any path $\PP_{\alpha}$ which contains a vertex which is marked by more than one vertex. This way, all remaining paths are going to be vertex disjoint from each other and hence correspond to independent random walks. 

We show that $\Mark$ algorithm can be implemented in parallel for all the vertices in $\FV^*_1$, and is used to identify all the independent random walks.

\begin{tbox}
${\ParallelMark:}$ An algorithm for detecting independent random walks for  $\FV^*_1$. 

\algline

			\begin{enumerate}
				\item Set $S_{\log t} = \emptyset$ initially.
				\item For every $\alpha \in \FV^*_1$ \textbf{in parallel} add $(\alpha, \alpha)$ to $S_{\log t}$.
				\item For $k =  \log t, \log t - 1, \dots, 1$: 
				\begin{enumerate}
					\item Set $S_{k-1} = \emptyset$ initially.
					\item For every $(\alpha, \beta) \in S_{k}$ \textbf{in parallel} add $(\alpha, \beta)$  to $S_{k - 1}$, and add $(\alpha, N_{k}(\beta))$ to $S_{k-1}$ if $N_k(\beta) \neq \perp$.
				\end{enumerate}
				\item Let $T$ be the set of $\beta$ such that there are $\alpha_1 \neq \alpha_2$ such that both $(\alpha_1, \beta)$ and $(\alpha_2, \beta) $ are in $S_0$ (by sorting all the pairs in $S_0$ according to the second coordinate).
				\item Return the set $\{\alpha : \nexists \beta \text{ s.t. } (\alpha, \beta) \in S_0, \beta \in T \}$.
			\end{enumerate}
\end{tbox}		
By the description of Algorithm $\ParallelMark$ and since sorting can be done in $O(1/\delta)$ rounds if memory per machine is $O(n^\delta)$, we obtain the following claim.
\begin{claim}\label{claim:detection_independent}
Algorithm $\ParallelMark$ returns a set of vertices in $\FV^*_1$ such that $\alpha$ is in the set iff $P_\alpha$ is an independent random walk for any $\alpha \in \FV^*_1$.

\noindent
Algorithm $\ParallelMark$  requires $O(t^2 \cdot n^{1-\delta})$ machines each with $O(n^{\delta})$ memory and  $O(\frac{1}{\delta} \cdot \log{t})$ MPC rounds.
\end{claim}

\begin{proof}[Proof of Theorem~\ref{thm:random-walk}]
	We simply run $\SRW(G,t)$ in parallel $\Theta(\log{n})$ times and detect the independent random walks found by each run using the marking procedure above. 
	By Lemma~\ref{lem:srw}, with probability $1/2$ we are able to find an independent random walk for any fixed vertex in each of the 
	$\Theta(\log{n})$ trials. Hence, with high probability, we are able to find an independent random walk for every vertex of $G$. This concludes the proof of correctness of the algorithm. 
	
	We now briefly describe the MPC implementation details of this algorithm. To implement $\SRW(G,t)$, we first create the vertex-set of the graph of $\FG(G,t)$ which consists of $O(n \cdot t^2)$ vertices. 
	We make every vertex responsible for maintaining the $O(\Delta \cdot t)$ of its neighbors and performing the random walks (the information needed by any single vertex resides entirely on one machine). 
	Sampling $\FGS$ is then straightforward. The rest of the algorithm can also be implemented in a straightforward way by spending $O(1/\delta)$ rounds for each iteration of for-loop in Line~(\ref{for:srw}) of $\SRW$. 
	By Claim~\ref{claim:detection_independent}, $\ParallelMark$ also needs $O(\log{t}/\delta)$ rounds. Hence, 
	in total, we only need $O(\log{t}/\delta)$ MPC rounds to implement the algorithms. 
	
	As for the memory per machine, for any fixed vertex, we only need $O(\Delta)$ (as opposed to $O(\Delta \cdot t)$) on the machine this vertex resides to sample an edge from $\FGS$ as the $O(\Delta \cdot t)$ neighbors of any
	vertex in $\FG(G,t)$ can be described by only $O(\Delta)$ edges (the rest are copies of the same edge to multiple copies of the same vertex on the next layer). 
	We further need to store $O(\log{t})$ intermediate vertices in $N(\cdot)$ and so each vertex needs
	$O(\Delta  + \log{t})$ memory and we have $O(n \cdot t^2 \cdot \log{n})$ vertices in total (recall that we are performing $O(\log{n})$ parallel random walks), finalizing the proof. 
	\end{proof}

\subsection{Proof of Lemma~\ref{lem:second-step}: The Randomization Step}

We now use Theorem~\ref{thm:random-walk} to prove Lemma~\ref{lem:second-step}. In Lemma~\ref{lem:second-step}, we need to perform lazy random walks, while Theorem~\ref{thm:random-walk} is performing random walks. 
However, this is quite easy to fix: we simply add $\Delta$ self-loops to every vertex of $G$. This makes the graph $2\Delta$ regular while ensuring that the distribution of a random walk in the new graph corresponds to a lazy random walk in the original graph. 
We use $\tG$ to refer to this new $2\Delta$-regular graph after adding the self-loops. We are now ready to construct the graph $H$ in Lemma~\ref{lem:second-step}. 

\begin{proof}[Proof of Lemma~\ref{lem:second-step}]
	We construct the graph $\tG$ as specified above and run algorithm in Theorem~\ref{thm:random-walk} on this graph for random walks of length $T$ for $k = 50\log{n}$ times in parallel. In the following, we condition on the 
	high probability event that the random walk algorithm succeeds. 
	
	The graph $H$ is defined as follows: $V(H) = V(\tG) = V(G)$; for 
	any $u \in V(H)$, connect $u$ to the $k$ vertices $v_{u,1},\ldots,v_{u,k}$ found by the random walk algorithm for $u \in V(\tG)$.  We now establish the desired properties of $H$. 
	
	Let $G_i$ be any connected component of $G$. Any vertex $u \in V(G_i)$ in $H$ is connected to $k$ vertices in $V(G_i)$ in $H$: this is because a lazy random walk starting from a vertex
	in $V(G_i)$ cannot ``escape'' the component $G_i$. As such, any vertex $u \in V(G_i)$ is connected to $k$ vertices in $V(G_i)$. Hence, the distribution of $H_i$ is a graph in which every vertex is connected
	to $k = 50\log{n}$ other vertices in $V(H_i)$ chosen according to the distribution of a lazy random walk of length $T$ in graph $G_i$. The distribution $\PG(n_i,100\log{n})$ is a distribution on which every vertex in $V(H_i)$ is connected
	to $(100\log{n}/2) = k$ vertices in $V(H_i)$ chosen uniformly at random. Since we are performing lazy random walks of length at least $T_{\gamstar}(G_i)$, we expect these two distributions to be close to each other. 
	
	Formally, let $\unif_{V(H_i)}$ denote the uniform distribution on $V(H_i)$. We have, 
	\begin{align*}
		\tvd{\distribution{H_i}}{\PG(n_i,100\log{n})} \leq \sum_{u \in V(\tG_i)}\tvd{\distrw(u,T)}{\unif_{V(H_i)}} \leq n_i \cdot {1}/{n^{10}} \leq 1/n^9. 
	\end{align*}
	This proves the second part of the lemma. To prove the first part of the lemma we need to prove that each $H_i$ is connected with high probability. This follows because
	$H_i$ has a similar distribution as $\PG(n_i,100\log{n})$ and a graph sampled from $\PG(n_i,100\log{n})$ is connected with probability at least $1-1/n^{25}$ by Proposition~\ref{prop:PG-connected} (by setting $d=100\log{n} \geq 100\log{n_i}$ and 
	assuming $n_i \geq 2$ as $G$ contains no isolated vertices), and hence by Fact~\ref{fact:tvd-small} $H_i$ is also connected with probability at least $1-1/n^{25}-1/n^{9}$, finalizing the proof of correctness.
	
	The number of machines needed by this algorithm is $O(\log{n})$ times the number of machines in Theorem~\ref{thm:random-walk} for $t = T$ and the memory per machine is the same. Hence the bounds on the MPC
	implementation of this algorithm follows immediately from Theorem~\ref{thm:random-walk}.   
\end{proof}

\section{Step 3: Connectivity on Random Graphs}\label{sec:finding-cc}

In this section we present the final and paramount step of our algorithm, which involves finding connected components 
of a collection of disjoint random graphs chosen from $\PG$.

\begin{mdframed}[hidealllines=false,backgroundcolor=gray!10,innertopmargin=0pt]
\begin{lemma}\label{lem:third-step}
	Let $G(V,E)$ be a graph on $n$ vertices such that any connected component $G_i$ of $G$ with $n_i = \card{V(G_i)}$ is sampled from $\PG(n_i,100\log{n})$. 
	There exists an MPC algorithm which identifies all connected components of $G$ with high probability (over both the randomness of the algorithm and the distribution $\PG$). 
	
	\noindent
	For any $\delta > 0$, the algorithm can be implemented with $O(n^{1-\delta}) \cdot \polylog{(n)}$ machines each with $O(n^{\delta}) \cdot \polylog{(n)}$ memory and $O(\frac{1}{\delta} \cdot \log\log{n})$ MPC rounds. 
\end{lemma}
\end{mdframed}

During the course of our exposition in this section, we need to set many parameters which we collect here for convenience.
\begin{tbox}
\vspace{-17pt}
\begin{align}
	\eps &:= \paren{100 \cdot \log{n}}^{-2} \textnormal{~: used to bound the discrepancy factor of almost-regular graphs,} \notag \\
	s &:= \frac{10^{6} \cdot \log{n}}{\eps^2} \textnormal{~: a scaling factor on degree of almost-regular graphs,} \notag \\
	\Delta &:= 100s \textnormal{~: used as a parameter to denote the degree of almost-regular graphs,} \notag \\
	F &:= \arg\min_{i}\set{\Delta^{2^i} \geq n^{1/100}} \textnormal{~: used to bound the number of phases in our algorithm.}\label{eq:parameters}
\end{align}
\end{tbox}

Throughout this section, we typically define the degree of almost-regular graphs by multiplicative factors of $s$; this is needed to simplify many concentration bounds used in the proofs. 
We further point out that $F = O(\log\log{n})$ and $\Delta^{F} \in [n^{1/100},n^{1/50}]$ and hence $\Delta^F = o(\eps)$. 

\paragraph{Preprocessing step.} The first step in proving Lemma~\ref{lem:third-step}, is to make each connected component $G_i$ of $G$ ``more random'', i.e., turn it to a graph sampled from $\PG$ with larger per-vertex degree. This can be 
easily done using Lemma~\ref{lem:second-step} in previous section, as the graph $G_i \sim \PG(n_i,100\log{n})$ has a small mixing time by Proposition~\ref{prop:PG-expansion} with high probability.

Now consider the following preprocessing process: Recall the parameters defined in Eq~(\ref{eq:parameters}). For $(F \cdot \Delta \cdot s/(100\log{n}))$ steps in parallel, we run the algorithm in Lemma~\ref{lem:second-step} on the original
graph $G$. For each connected component $G_i$ of $G$, this results us in having $F$ graphs $\tG_{i,1},\ldots,\tG_{i,F}$ which are (almost) sampled from the distribution $\PG(n_i,\Delta \cdot s)$ (the distribution of these graphs is not exactly identical to this, but 
is rather close to this distribution in total variation distance which is sufficient for our purpose). As such, we now need to find the connected component of a graph $\tG$ which is the union of all $\tG_{i,j}$ for $i$ ranging over all connected
components of $G$ and $j \in [F]$. 

In the following lemma, we design an algorithm for this task. For simplicity of exposition, we state this lemma for the case of finding a spanning tree of one such connected component (i.e., assuming $G$ itself is sampled from $\PG$ as opposed to having its connected components sampled from this distribution); however, it would be evident that running this algorithm on the original input results in finding a spanning tree of each connected component separately. 

\begin{lemma}\label{lem:technical-random-cc}
	Let $\tG$ be a graph on $n$ vertices such that $\tG = \tG_1 \cup \ldots \cup \tG_F$ where $\tG_i \sim \PG(n,\Delta \cdot s)$. 
	There exists an MPC algorithm that can find a spanning tree of $\tG$ with high probability (over both the randomness of the algorithm and the distribution $\PG$). 
	
	\noindent
	For any $\delta > 0$, the algorithm can be implemented with $O(n^{1-\delta}) \cdot \polylog{(n)}$ machines each with $O(n^{\delta}) \cdot \polylog{(n)}$ memory, and in $O(\frac{1}{\delta} \cdot \log\log{n})$ MPC rounds. 
\end{lemma}

\noindent
We note that in Lemma~\ref{lem:technical-random-cc}, the input to the algorithm is the collection of graphs $(\tG_1,\ldots,\tG_F)$ (i.e., the algorithm knows partitioning of $G$ into its $F$ subgraphs; think of each input edge being labeled by the graph $\tG_i$
it belongs to). The rest of this section is devoted to the proof of Lemma~\ref{lem:technical-random-cc}. At the end of the section, we use this lemma to prove Lemma~\ref{lem:third-step}. In this section, $n$ always
refer to number of vertices in $\tG$.

\subsection{Proof of Lemma~\ref{lem:technical-random-cc}: Connectivity on a Single Random Graph}

We start by defining a natural operation on graphs in context of connectivity. 

\begin{definition}[Contraction Graph]\label{def:contraction-graph} 
For a graph $G(V,E)$ and a partition $\CC:=\set{C_1,\ldots,C_k}$ of $V(G)$ (not necessarily a component-partition), we construct a \emph{contraction graph} $H$ of $G$ with respect to $\CC$ as the following graph: 
\begin{enumerate}
	\item \textbf{Vertex-set:} The vertex-set $V(H)$ of $H$ is a collection of $k$ vertices where $w_i \in V(H)$ is labeled with the component $C_i$ of $\CC$, denoted by $C(w_i)$. 
	\item \textbf{Edge-set:} For any $w \neq z \in V(H)$, there exists an edge $(w,z) \in E(H)$ iff there exists vertices $u \in C(w)$ and $v \in C(z)$ where $(u,v) \in E(G)$ ($H$ contains no parallel edges and no self-loops). 
\end{enumerate}
In other words, $H$ is obtained by ``contracting'' the vertices of $G$ inside each set of $\CC$ into a single vertex and removing parallel edges and self-loops. 
\end{definition}

\noindent
Suppose $\CC$ is a component-partition of $G$ and $H$ is a contraction graph of $G$ with respect to $H$. Then it is immediate to see that we can construct a spanning tree (or forest) of $G$ given only spanning trees of each
component in $\CC$ and a spanning tree of $H$. 

\paragraph{Overview of the algorithm.} The algorithm in Lemma~\ref{lem:technical-random-cc} goes through $F$ \emph{phases}. In each phase $i \in [F]$, it only considers the edges in $\tG_i$ and use them to ``grow'' the components of $\tG$ found in
the previous phases. This part is done using a new \emph{leader-election} algorithm that we design in this paper. This algorithm takes the contraction graph of $\tG_i$ with respect to the set of components found already, and merge
these components further to build larger components. The novelty of this leader-election algorithm is that starting from an (almost) $d$-regular graph, it can grow each component by a factor of (almost) $d$ (as opposed to typical leader-election algorithms that 
only increase size of each component by a constant factor). 

Our main algorithm is then obtained by successively applying this leader election algorithm to contraction graph of $\tG_i$ to build relatively large components of $\tG_i$
and use them to refine the components found for $G$. The main step of our proof is to argue that if contraction graph of $\tG_i$ was a \emph{random} (almost) $d$-regular graph on $n'$ vertices, then
the contraction graph of $\tG_{i+1}$ in this process would be another random (almost) $d^2$-regular graph on roughly $n'/d$ vertices. Having achieved this, we can argue that each component of the graph $G$ grows by a quadratic factor in each phase, and
 hence after only $O(\log\log{n})$ phase, each component has size $n^{\Omega(1)}$ (due to technical reasons, one cannot continue this argument until just one connected component of size $n$ remains). 
 Finally, we prove that at this step, the diameter of the remaining graph, i.e., contraction of $G$ on the 
found components is only $O(1)$. A simple broadcasting algorithm can then be used to found a spanning tree of the remaining graph in $O(1)$ rounds.

\subsection*{A Leader Election Algorithm}

We first introduce a simple leader election algorithm, called $\LE(H,d)$, which gets as an input an (almost) ($d \cdot s$)-regular graph and creates  
components of size (almost) $d$ in this graph. We note that the description of the algorithm
itself does not depend on the fact that $H$ is almost-regular.

\begin{tbox}
	$\LE(H,d)$. A simple leader election algorithm for growing connected components on an (almost) $(d \cdot s)$-regular graph $H$. 
	
	\algline	 
	
	\begin{enumerate}
		\item Set $L = \emptyset$ initially.
		\item For every vertex $v \in V(H)$ \textbf{in parallel}
		independently sample $p(v)$ from the Bernoulli distribution with probability $p:= s/d$
		and insert $u$ to $L$ iff $p(v) = 1$ (we refer to these vertices as \emph{leaders}).
		\item Let $R:= V(H) \setminus L$. 
		\item For any vertex $v \in R$ \textbf{in parallel} set $N_{L}(v)$ be the set of neighbors of $v$ in $L$ in graph $H$. 
		\item For any vertex $v \in R$ \textbf{in parallel} let $M(v)$ be a vertex $u \in R$ chosen uniformly at random from $N_L(v)$ (we define $M(v) = \perp$ if $N_{L}(v) = \emptyset$).
		\item Return $k := \card{L}$ sets $S_{v_1},\ldots,S_{v_k}$ 
		for $v_1, \ldots, v_k \in L$ such that 
		$S_{v_i} = \{v_i\} \cup \{u \in R : M(u) = v_i\}$ (vertices with $M(u) = \perp$ are ignored).
	\end{enumerate}
\end{tbox}

We have the following immediate claim.  
\begin{claim}\label{clm:LE-simple}
	Suppose in $\LE$ no vertex $v \in R$ has $M(v) = \perp$. Then, the returned collection $S_1,\ldots,S_k$ is a component-partition of $H$.  
\end{claim}
\begin{proof}
	The induced graph of $H$ on any set $S_i$ contains a star with the leader in $S_i$ being the its center. Hence, each $S_i$ is a component of $H$. Moreover, by definition, the sets $S_i$'s are disjoint. 
	Finally, since for no vertex $v \in R$, $M(v) = \perp$, $S_i$'s contain all vertices in $H$. 
\end{proof}

The main property of \LE is that when computed on almost regular graphs it results in a component-partition with \emph{almost equal size components}. In other words,
if $H$ is a $\range{(1\pm \eps) d \cdot s}$-almost-regular graph, then the resulting components are of size $\range{\paren{1\pm O(\eps)} \cdot d}$ each. 
\begin{lemma}[Equipartition Lemma]\label{lem:LE-equal-size}
	Let $\beps \in (\eps,1/100)$ and $H$ be a $\range{(1\pm \beps)~d \cdot s}$-almost-regular graph for $d \geq s$. 
	Then, with probability $1-1/n^{23}$, for $(S_1,\ldots,S_k) = \LE(H,d)$: 
	\begin{enumerate}
	\item For all $i \in [k]$, $\card{S_i} \in \range{\paren{1 \pm 3\beps} d}$,
	\item $(S_1,\ldots,S_k)$ is a component-partition of $V(H)$. 
	\end{enumerate}
\end{lemma}
\begin{proof}
	Define $\eps' = \beps/10$ and so $s \geq 100\log{n}/\eps'^2$ by Eq~(\ref{eq:parameters}). 
	Throughout the proof, we repeatedly use the facts that $\range{\paren{1 \pm \eps'}^{-1}} \subseteq \range{\paren{1 \pm 2\eps'}}$ and $\range{\paren{1 \pm \eps'}^2} \subseteq \range{\paren{1 \pm 3\eps'}}$ as $\eps' =o(1)$.
	
	Fix any vertex $u \in R$ and let $d_u \in \range{\paren{1\pm 10\eps'} d \cdot s}$ be the degree of $u$ in $H$. We define $d_u$ random variables $X_1,\ldots,X_{d_u}$ where $X_i = 1$ iff the $i$-th neighbor of $u$ is chosen as a leader in $L$ and 
	$X_i = 0$ otherwise. Let $X = \sum_{i}X_i$ denote the number of neighbors of $u$ in $L$. As the choice of any leader is independent of whether $u$ belongs to $L$ or not, we have $\Ex\bracket{X} = d_u \cdot p \in \range{(1\pm10\eps')s}$. 
	Moreover, by Chernoff bound,
	\begin{align*}
		\Pr\paren{X \notin \range{\paren{1\pm \eps'} \Ex\bracket{X}}} &\leq \exp\paren{-\frac{\eps'^2 \cdot d_u \cdot p}{2}} \leq \exp\paren{-\frac{\eps'^2 (1-10\eps') \cdot s}{2}} \\
		&\leq \exp\paren{-25\log{n}} \leq \frac{1}{n^{25}} \tag{as $s \geq 100\log{n}/\eps'^2$ and $\eps' =o(1)$}. 
	\end{align*}
	Consequently, w.p. $1-1/n^{25}$, $\card{N_L(u)} \in \range{\paren{1 \pm \eps'} \cdot \paren{1 \pm 10\eps'} \cdot s} \subseteq \range{\paren{1 \pm 12\eps'}\cdot s}$ (as $\eps' =o(1)$).
	By union bound, this event happens for all vertices in $R$ w.p. $1-1/n^{24}$. In the following, we condition on this event. The second part of the lemma already follows from this and Claim~\ref{clm:LE-simple}. 
	
	Now fix a vertex $v \in L$. Define $N_R(v)$ as the set of neighbors of $v$ in set $R$ in graph $H$. The same exact argument as above implies that with probability $1-1/n^{24}$, for all vertices in $L$, $\card{N_R(v)} \in \range{\paren{1 \pm12\eps'}d \cdot s}$. 
	We further condition
	on this event. 
	
	Consider again a vertex $v \in L$. For any vertex $u \in N_R(v)$, we define a random variable $Y_u$ where $Y_u = 1$ iff $M(u) = v$, i.e., $u$ chooses $v$ as its leader. Define $Y = \sum_{u} Y_u$. We point out that $Y+1$ is the size of 
	component returned by $\LE$ which contains the leader $v$. Hence, it suffices to bound $Y$ to finalize the proof. We have, 
	\begin{align*}
		\Ex\bracket{Y} = \sum_{u \in N_R(v)} \Ex\bracket{Y_u} = \sum_{u \in N_R(v)} \frac{1}{\card{N_L(u)}} \in  \range{\frac{\paren{1 \pm 12\eps'} d \cdot s}{\paren{1 \pm 12\eps'} s}} \subseteq \range{\paren{1\pm25\eps'} \cdot d},
	\end{align*}
	as $\card{N_R(v)} \in \range{\paren{1 \pm12\eps'}d \cdot s}$ and $\card{N_L(u)} = \range{\paren{1 \pm 12\eps'}\cdot s}$ and $\eps' =o(1)$. By Chernoff bound, 
	\begin{align*}
		\Pr\paren{Y \notin \range{\paren{1\pm \eps'} \Ex\bracket{Y}}} \leq \exp\paren{-\frac{\eps'^2 \cdot (1-25\eps') \cdot d}{2}} \leq \exp\paren{-25\log{n}} \leq \frac{1}{n^{25}}. 
	\end{align*}
	A union bound on all vertices in $L$ implies that $\card{S_i} \in \range{\paren{1 \pm 27\eps'} d} \subseteq \range{\paren{1 \pm 30\eps'} d}$ with probability $1-1/n^{24}$. Taking another union bound on all the events 
	conditioned on in the proof, with probability $1-1/n^{23}$, we obtain that $\card{S_i} \in \range{\paren{1 \pm 30\eps'} d} = \range{\paren{1 \pm 3\beps} d}$, finalizing the proof. 	
\end{proof}

We have the following claim by the definition of Algorithm $\LE$.
\begin{claim}\label{claim:leader_election}
Algorithm $\LE(H,d)$  requires $O(|E(H)| / n^\delta)$ machines each with $O(n^{\delta})$ memory and  $O(1/\delta)$  MPC rounds.
\end{claim}

\subsection*{Growing Connected Components} 

We now use \LE algorithm from the previous section to design our main algorithm which ``grows'' the size of connected components of $G$ repeatedly over $F$ phases.

\begin{tbox}
	$\GC(\tG,\Delta)$. An algorithm for ``growing'' connected components of size up to $n^{\Omega(1)}$ in a given graph $\tG = \tG_1 \cup \ldots \cup \tG_F$ where $\tG_i \sim \PG(n,\Delta \cdot s)$. 

	\algline	 
	\begin{enumerate}
		\item Let $\CC_1$ be a partition of $V(\tG)$ into singleton sets. 
		\item For $i=1$ to $F$ phases: 
		\begin{enumerate}
			\item Let $\Delta_i := \Delta^{2^{i-1}}$ and $p_i = \Delta_i^{-1} \cdot s$. 
			\item For every vertex $v \in V(\tG_i)$ \textbf{in parallel} let $c_i(v) = j$ for $v \in C_j$.
			\item Construct contraction graph $H_i$ of $\tG_i$ (\emph{not} $\tG$) with respect to $\CC_i$ as follows:
			\begin{enumerate}
				\item Set $H_i$ to be an empty set inititially. 
				\item For every edge $(u, v) \in E(\tG_i)$ \textbf{in parallel} add $(C_{c_i(u)}, C_{c_i(v)})$ to $H_i$.
			\end{enumerate} 
			\item Compute $(S_1,\ldots,S_k) = \LE(H_i,\Delta_i)$ (hence, each $S_j \subseteq V(H_i)$). 
			\item For each $S_j$ \textbf{in parallel} let   $C_{i+1,j} = \bigcup_{w \in S_j} C_{i}(w)$.
			\item Let $\CC_{i+1} = \set{C_{i+1,1},\ldots,C_{i+1,k}}$.
		\end{enumerate}
		\item Return the graph $H_F$. 
	\end{enumerate}
\end{tbox}

The following claim is straightforward from the description of  $\GC$ and Claim~\ref{claim:leader_election}.
\begin{claim}\label{clm:grow_components}
Algorithm $\GC(\tG,\Delta)$  requires $O(|E(\tG)| / n^\delta)$ machines each with $O(n^{\delta})$ memory and  $O(F/\delta)$  MPC rounds.
\end{claim}

We prove that for each phase $i \in [F]$, the contraction graph $H_i$ constructed in this phase is an almost-regular graph with degree roughly $\Delta_i \cdot s$ and discrepancy factor $\eps_i := \paren{20^{i} \cdot \eps}$. 
The following lemma is the heart of the proof. 

\begin{lemma}\label{lem:grow-induction}
	In $\GC(\tG,\Delta)$, with high probability, for any $i \in [F]$: 
	\begin{enumerate}[label=(\Roman*)]
		\item	$\CC_i$ is a component-partition of $\tG$ with $\card{C_{i,j}} \in \range{\paren{1 \pm \eps_i}\Delta_i /\Delta}$ for all $C_{i,j} \in \CC_i$.
		\item $H_i$ is a $\range{\paren{1 \pm \eps_i} \Delta_i \cdot s}$-almost-regular graph on $n_i \in \range{\paren{1 \pm \eps_i} \cdot {n \Delta}/{\Delta_i}}$ vertices. 
	\end{enumerate}
\end{lemma}
\begin{proof}
	We prove this lemma inductively. 
	
	\paragraph{Base case:} $\CC_1$ is clearly a component-partition of $\tG$ as it only consists of singleton sets and $\card{C_{1,j}} = 1$ for all $C_{1,j} \in \CC_1$. Since $\Delta_1 = \Delta$, this proves the first part of the lemma in the base case. 
	For the second part, as $\CC_1$ only consists of singleton sets, $H_1 = \tG_1$ and hence $n_1 = n$. Finally, $H_1 = \tG_1 \sim \PG(n,\Delta \cdot s)$ and hence by Proposition~\ref{prop:PG-regular} (as $s \geq 100\log{n}/\eps^2$), $H_1$ is a
	$\range{\paren{1 \pm \eps} \Delta \cdot s}$-almost-regular graph, hence concluding the proof of the base case.
	
	\paragraph{Induction step:} Now suppose this is the case for some $i > 1$ and we prove it for $i+1$. 
	By induction, we have that $H_i$ is a $\range{\paren{1 \pm \eps_i} \Delta_i\cdot s} $-almost-regular graph on $n_i \in \range{\paren{1 \pm \eps_i} \cdot {n \cdot \Delta}/{\Delta_i}}$ vertices. In this phase, 
	we compute $(S_1,\ldots,S_k) = \LE(H_i,\Delta_i)$.  We can thus apply Lemma~\ref{lem:LE-equal-size} with parameters $d= \Delta_i$, $p = p_i$, and $\beps =  \eps_i < 1/100$, and obtain that
	with high probability, 
	\begin{align}
	\card{S_i} \in \range{\paren{1 \pm 3\beps} \cdot \Delta_i} = \range{\paren{1 \pm 3\eps_i} \cdot \Delta_i}, \label{eq:card-Si}
	\end{align}
	and $(S_1,\ldots,S_k)$ is a component-partition of $H_i$. In the following, we condition on this event. 
	
	\paragraph{Proof of part (I):} Since $\CC_i$ is a component-partition of $\tG$ (by induction), we have that vertices in $H_i$ correspond to components of $\tG$, i.e., vertices in $C_i(w)$ for all $w \in V(H_i)$ are connected in $\tG$. 
	Moreover, by Lemma~\ref{lem:LE-equal-size}, $(S_1,\ldots,S_k)$ is a component-partition of $H_i$ and hence vertices (of $H_i$) in each $S_j$ for $j \in [k]$ are connected to each other (in $H_i$). 
	As edges of $H_i$ correspond to edges in $\tG_i \subseteq \tG$, any $C_{i+1,j} \in \CC_{i+1}$ is a component of $\tG$, hence $\CC_{i+1}$ is a component-partition of $\tG$. 
	
	We now prove the bound on size of each $C_{i+1,j} \in \CC_{i+1}$. By definition, 
	\begin{align}
	\card{C_{i+1,j}} &= \sum_{w \in S_j} \card{C_{i}(w)} \in \range{\card{S_j} \cdot \paren{1 \pm \eps_i}\Delta_i/\Delta} \tag{by induction hypothesis on ${C_i(w)} \in \CC_i$} \notag \\ 
	&\subseteq \range{(\paren{1 \pm 3\eps_i} \cdot \Delta_i) \cdot (\paren{1 \pm \eps_i}\Delta_i/\Delta)} \tag{by Eq~(\ref{eq:card-Si})} \notag \\
	&\subseteq \range{\paren{1 \pm 5\eps_i} \cdot \Delta_i^2/\Delta} = \range{\paren{1 \pm 5\eps_i} \cdot \Delta_{i+1}/\Delta},  \label{eq:sharper-C}
	\end{align}
	as $\Delta_i^2 = \Delta_{i+1}$. By the choice of $\eps_{i+1} > 5\eps_{i}$, this finalizes the proof of the first part. We now consider the second part.

	\paragraph{Proof of part (II):} Notice that $n_{i+1} = \card{\CC_{i+1}}$ as each set in $\CC_{i+1}$ is contracted to a single vertex in $H_{i+1}$. Since $\CC_{i+1}$ partitions $V(G)$, and 
	as by Eq~(\ref{eq:sharper-C}) each set in $\CC_{i+1}$ has size in $\range{\paren{1 \pm 5\eps_i} \cdot \Delta_{i+1}/\Delta}$, 
	we have 
	\begin{align}
		n_{i+1} \in \range{\frac{n}{\paren{1 \pm 5\eps_i} \cdot \Delta_{i+1}/\Delta}} \subseteq \range{\paren{1 \pm 6\eps_i} {n \cdot \Delta}/{\Delta_{i+1}}}. \label{eq:sharper-n}  
	\end{align}
	As $\eps_{i+1} > 6\eps_{i}$, this proves the bound on $n_{i+1}$. It remains to prove $H_{i+1}$ is an $\range{\paren{1 \pm \eps_{i+1}} \Delta_{i+1} \cdot s}$-almost-regular graph. This is the main part of the argument. 
	\begin{lemma}\label{lem:dw}
		For any vertex $w \in V(H_{i+1})$, degree of $w$ in $H_{i+1}$ is $d_w \in \range{\paren{1 \pm \eps_{i+1}} \Delta_{i+1} \cdot s}$ with high probability. 
	\end{lemma}
	\begin{proof}
	Recall that $H_{i+1}$ is a contraction graph of $\tG_{i+1}$ with respect to the partition $\CC_{i+1}$. We define $C=C_{i+1}(w) \in \CC_{i+1}$.
	In $H_{i+1}$, $w$ has 
	an edge to another vertex $z \in V(H_{i+1})$ iff there exists a vertex $u \in C \subseteq V(\tG_{i+1})$ such that $u$ has an edge to some vertex $v \in C_{i+1}(z)$ in the graph $\tG_{i+1}$. 
	As such, degree of $w$ is equal to the number of sets $C_{i+1,j} \subseteq V(\tG_{i+1})$ such that there is an edge $(u,v) \in E(\tG_{i+1})$ for $u \in C$ and $v \in C_{i+1,j}$. 
	
	Now consider the process of generating $\tG_{i+1} \sim \PG(n,\Delta \cdot s)$ and notice that the edges chosen in $\tG_{i+1}$ are chosen independent of the choice of $\CC_{i+1}$ as $\CC_{i+1}$ is only a function
	of the graphs $\tG_1,\ldots,\tG_i$. Moreover, recall that 
	in $\PG(n,\Delta \cdot s)$ each vertex chooses $\Delta \cdot s/2$ other vertices uniformly at random to connect to (and then we remove the direction of edges). For any two sets $S,T \subseteq V(\tG_{i+1})$, 
	we say that $S$ ``hits'' $T$ if there exists a vertex in $S$ which picks a directed edge to some vertex in $T$ in the process of generating $\tG_{i+1}$ (so it is possible that $S$ hits $T$ but $T$ does not hit $S$).
	Let $K \subseteq [k]$ be such that for each $j \in K$, either $C$ hits $C_{i+1,j}$ or vice versa. By the above argument $d_w \in \range{\card{K} \pm 1}$ (to account for the fact that $C$ hitting $C$ does not change the degree of $w$ as we have no self-loops in $H_{i+1}$). 
	In the following two claims, we bound $\card{K}$. 
	\begin{claim}\label{clm:K-out}
		Let $K^+ \subseteq [k]$ be the set of all indices $j \in [k]$ such that $C$ hits $C_{i+1,j}$. Then, with high probability, $\card{K^+} \in \range{\paren{1 \pm \eps_{i+1}} \Delta_{i+1} \cdot s/2}$. 
	\end{claim}
	\begin{proof}
	We model the number of sets hit by $C$ as a balls and bins expriment (see Appendix~\ref{app:balls-and-bins}): ``balls'' are the edges going out of vertices in $C$ in construction of $\tG_{i+1}$ in $\PG$ and ``bins'' are the sets $C_{i+1,j}$ for $j \in [k]$. 
	Hence, non-empty bins are exactly the 
	set $K^+$ and thus it suffices to bound number of non-empty bins.
	
	In $\tG_{i+1}$, any vertex in $C$ is choosing $\Delta \cdot s/2$ directed edges. As such, the total number of balls in this argument is $N = \card{C} \cdot \Delta \cdot s/2 \in \range{\paren{1 \pm 5\eps_i}\Delta_{i+1} \cdot s/2}$
	by the bound proven on $\card{C}$ in Eq~(\ref{eq:sharper-C}).
	
	The total number of bins in this argument is $B = \card{\CC_{i+1}} = n_{i+1} \in \range{\paren{1 \pm 6 \eps_i} {n \cdot \Delta}/{\Delta_{i+1}}}$ as proven in Eq~(\ref{eq:sharper-n}). Moreover, 
	for any $j \in [k]$, $\card{C_{i+1,j}} \in \range{\paren{1 \pm 5 \eps_i}\Delta_{i+1}/\Delta}$ (as stated above for $C$). As such, the ratio between the largest and smallest set in $\CC_{i+1}$ is in
	$\range{\paren{1 \pm 10\eps_{i}}}$. Moreover, edges going out of $C$ are chosen uniformly at random from, and hence each bin in this argument is chosen 
	with probability in $\range{\paren{1 \pm 10\eps_{i}} \cdot B^{-1}}$. Furthermore,
	\begin{align*}
		\frac{B}{N} &\in \range{\frac{\paren{1 \pm 6\eps_i} {n \cdot \Delta}/{\Delta_{i+1}}}{{\paren{1 \pm 5\eps_i}\Delta_{i+1} \cdot s/2}}} \implies 
		\frac{B}{N} \geq \frac{ n \cdot \Delta}{2\Delta_{F}^2} \gg \polylog{(n)} \gg \frac{1}{10\eps_i}, 
	\end{align*}
	where the inequalities are by choice of $F$ and $\eps$ because $\eps_F = o(1)$ and $\Delta_{F} = \Delta^{2^{F}} \leq n^{1/50}$. 
 	
	Let $X$ be the number of non-empty bins in this process. By Proposition~\ref{prop:balls-and-bins} for this balls and bins experiment: 
	\begin{align*}
		\Pr\Paren{X \notin \range{\paren{1\pm 20\eps_i} \cdot N}} \leq\exp\paren{-\frac{100\eps_i^2 \cdot N}{2}} = 1/n^{\omega(1)}. 
	\end{align*}
	Hence, with high probability, the total number of non-empty bins is in $\range{\paren{1 \pm 20\eps_i}\Delta_{i+1} \cdot s/2}$, which finalizes the proof as $\eps_{i+1} = 20\eps_i$. \Qed{Claim~\ref{clm:K-out}}
	
	\end{proof}
	One interpretation of Claim~\ref{clm:K-out} is that  that distribution of $H_{i+1}$ is similar to $\PG(n_{i+1},\Delta_{i+1} \cdot s)$ with the difference that the number of out-edges chosen in $\PG$ is not exactly
	$\Delta_{i+1} \cdot s$ (but quite close to it for each vertex). As such, we would expect $H_{i+1}$ to still behave similarly as $\PG(n_{i+1},\Delta_{i+1} \cdot s)$; in particular, be almost-regular with
	high probability. The following claim is analog of Proposition~\ref{prop:PG-regular} for the distribution of $H_{i+1}$.  

	\begin{claim}\label{clm:K-in}
		Let $K^- \subseteq [k]$ be the set of all indices $j \in [k] \setminus K^+$ such that $C_{i+1,j}$ hits $C$. Then, with high probability, $\card{K^-} \in \range{\paren{1 \pm \eps_{i+1}} \Delta_{i+1} \cdot s/2}$. 
	\end{claim}
	\begin{proof}
		Fix any $j \in [k] \setminus K^+$. As shown in Claim~\ref{clm:K-out}, $c_j := \card{C_{i+1,j}} \cdot \Delta \cdot s/2 \in \range{\paren{1 \pm 5\eps_i}\Delta_{i+1} \cdot s/2}$ edges are going out 
		of vertices in $C_{i+1,j}$. Any such edge, would hit the set $C$ with probability $p:= \card{C}/n$. Let $\eps' = (1/\log{n})^{10} \ll \eps$. We have,
		\begin{align*}
			\Pr\paren{C_{i+1,j}~\textnormal{hits}~C} = 1-\paren{1-p}^{c_j} \in \range{(1 \pm c_j p) \cdot c_j p} \subseteq \range{(1 \pm \eps') \cdot c_j p}. \tag{as $(1-x) \leq e^{-x} \leq 1-x+x^2$ and $c_j \cdot p = \Ot(\Delta_F^2)/n = \Ot(n^{2/50})/n \ll \eps'$}
		\end{align*}
		Moreover, we have $k = n_{i+1} \in \range{\paren{1 \pm 6 \eps_i} {n \cdot \Delta}/{\Delta_{i+1}}}$ and $\card{K^+} = \range{\paren{1 \pm \eps_{i+1}} \Delta_{i+1} \cdot s/2}$ and since $\Ot(\Delta_F^2)/n \ll \eps'$, we have
		that $\card{[k] \setminus K^+} \in \range{(1\pm \eps') \cdot \card{k}}$. 
		
		Let $X = \card{K^-}$ denotes the number of sets $C_{i+1,j} \in [k] \setminus K^+$ that hit $C$. By the above calculation: 
		\begin{align*}
			\Ex\bracket{X} &\in \range{(1 \pm \eps') \cdot \sum_{j \in [k] \setminus K^+} c_j p} \subseteq \range{(1 \pm \eps') \cdot (1\pm \eps') \cdot \card{k}\cdot (\paren{1 \pm 6 \cdot \eps_{i+1}}\Delta_{i+1})^2 \cdot s/2 \Delta n} 
			\tag{as $p=\card{C}/n$ and $\card{C} \in  \range{\paren{1 \pm 5\eps_i}\Delta_{i+1}/\Delta}$ and $c_j \in \range{\paren{1 \pm  5\eps_i}\Delta_{i+1} \cdot s/2}$} \\
			&\subseteq \range{(1 \pm 3\eps') \cdot \paren{1 \pm 13 \cdot \eps_i}\Delta_{i+1} \cdot s/2}. 
		\end{align*}
		Finally, notice that $X$ is a sum of independent random variables and hence by Chernoff bound,
		\begin{align*}
			\Pr\paren{X \notin \range{\paren{1\pm \eps_{i}} \Ex\bracket{X}}} \leq \exp\paren{-\frac{\eps^2 \cdot s}{4}}	 \leq \frac{1}{n^{25}}.	
		\end{align*}
		This implies that with high probability $X \in \range{\paren{1 \pm 15\eps_i}\Delta_{i+1} \cdot s/2}$. As $\eps_{i+1} > 15\eps_i$, this concludes the proof. \Qed{Claim~\ref{clm:K-in}}

	\end{proof}
	Lemma~\ref{lem:dw} now follows immediately from Claims~\ref{clm:K-out} and~\ref{clm:K-in}. \Qed{Lemma~\ref{lem:dw}}
	
	\end{proof}
	To conclude the proof of Lemma~\ref{lem:grow-induction}, we simply take a union bound on all vertices $w \in V(H_{i+1})$ and by Lemma~\ref{lem:dw} obtain that with high probability $d_w \in \range{\paren{1 \pm \eps_{i+1}}\Delta_{i+1} \cdot s}$.
	This implies that $H_{i+1}$ is a $\range{\paren{1 \pm \eps_{i+1}}\Delta_{i+1} \cdot s}$-almost regular graph, proving the induction step. \Qed{Lemma~\ref{lem:grow-induction}}

\end{proof}

We also state the following corollary of Lemma~\ref{lem:grow-induction} which roughly speaking, states that 
each graph $H_i$ is sampled from a distribution which is in spirit of $\PG(n_i,\Delta_i \cdot s)$ (with additional ``noise''). 
\begin{proposition}\label{prop:H_i-random}
	With high probability, distribution of each graph $H_{i}$ in \GC is a graph on $n_i \in \range{\paren{1 \pm \eps_{i}} \cdot {n \Delta}/{\Delta_i}}$ vertices in which we first pick 
	$\range{\paren{1 \pm \eps_i} \Delta_i \cdot s/2}$ many neighbors for each vertex where the other endpoint is chosen with probability in $\range{\paren{1 \pm 2\eps_i} \cdot n_{i}^{-1}}$
	and then remove the direction of edges. 
\end{proposition}

The proof of this proposition is identical to the proof of Claim~\ref{clm:K-out} using the fact that $H_i$ is almost-regular by Lemma~\ref{lem:grow-induction} (see also the discussion after Claim~\ref{clm:K-out}).  

Finally, we claim that $\GC$ can also find a spanning tree of components in $\CC_{F}$. 

\begin{claim}\label{clm:GC-spanning-tree}
	Let $T$ be the set of edges chosen in executions of $\LE$ (in defining $M(\cdot)$ for each non-leader vertex) in the course of execution of $\GC(\tG,\Delta)$. 
	With high probability, the induced subgraph of $T$ on each component in $\CC_F$ is a spanning tree.   
\end{claim}
\begin{proof}
	Follows immediately from Claim~\ref{clm:LE-simple} and the fact that each $\CC_{i+1}$ is formed by merging already found components of $\tG$ (see also the discussion after Definition~\ref{def:contraction-graph}).
\end{proof}
\subsection*{Building the Spanning Tree}

Recall that by the choice of $F$, $\Delta_F \in [n^{1/100},n^{1/50}]$ . By running $\GC(\tG,\Delta)$, we obtain a graph $H_F$ which consists of $n_F \in \range{\paren{1 \pm o(1)} \cdot n \cdot \Delta/\Delta_F}$ vertices (as $\eps_F = o(1)$ by Eq~(\ref{eq:parameters})). Additionally, by Proposition~\ref{prop:H_i-random}, with high probability, $H_F$ is a ``random'' graph (in spirit of $\PG$) with $\range{\paren{1 \pm o(1)} \Delta_F}$ ``out-degree'' before removing the direction of edges. 
We use this to bound the diameter of $H_F$. 
\begin{claim}\label{clm:H_F-diameter}
	Diameter of $H_F$ is $D = O(1)$ with high probability. 
\end{claim}
\begin{proof}
	We condition on the event that distribution of $H_F$ is as stated in Proposition~\ref{prop:H_i-random}. We argue that 
	for any set $S \subseteq V(H_F)$, the neighborset $N(S)$ of $S$ in $H_F$ has size 
	\begin{align*}
	\card{N(S)} \geq \min\set{2n_F/3, \Delta_F \cdot \card{S}/20}.
	\end{align*}
	The proof of this claim is exactly as in proof of Proposition~\ref{prop:PG-expansion}, using the analogy between $\PG$ and distribution of $H_F$ (with a minor additional care to account for the ``noise'' in $H_F$). We omit the details of this proof.  
	
	By the above equation, the $k$-hop neighborhood of any vertex in $H_F$ contains either at least $2n_F/3$ or $\paren{\Delta_F/20}^{k}$ vertices. In particular, for $k' = \log_{(\Delta_F/20)}{n}$, the $k'$-hop neighborhood of any vertex 
	contains at least $2n_F/3$ vertices. This implies that the $2k'$-hop neighborhood of any vertex contains the whole graph, hence the diameter of $H_F$ is $O(k')$. Since $\Delta_F \in [n^{1/100},n^{1/50}]$, we 
	obtain that diameter of $H_F$ is $O(1)$.
\end{proof}

We use the above claim to design a very simple algorithm to build a spanning tree of $H_F$. We then combine this algorithm with $\GC$ to prove Lemma~\ref{lem:technical-random-cc}.

\begin{claim}\label{clm:low-diameter}
	Let $H$ be any graph with $m$ edges, $n$ vertices, and diameter $D$. A spanning tree of $H$ can be found in $O(D/\delta)$ MPC rounds with $O(m^{1-\delta})$ machines with memory $O(m^{\delta})$ for any $\delta > 0$. 
\end{claim}
\begin{proof}
	We pick any arbitrary vertex $v \in H$. The algorithm proceeds in $D$ iterations. In the first iteration, $v$ informs all its neighbors in $H$ and add these edges to the underlying spanning tree. In the next iteration, the neighbors of $v$
	inform all their neighbors; any vertex informed which has already not chosen an edge in the spanning tree would pick one of its incoming edges and add it to the spanning tree. We continue like this until after $D$ iterations all vertices have a neighboring
	edge in the spanning tree. 
	
	It is straightforward that one can implement this algorithm in $O(D/\delta)$ MPC rounds on machines of memory $O(m^{\delta})$, hence finalizing the proof. 
\end{proof}

We are now ready to conclude the proof of Lemma~\ref{lem:technical-random-cc}. 

\begin{proof}[Proof of Lemma~\ref{lem:technical-random-cc}]
	By Claim~\ref{clm:GC-spanning-tree}, we can find a spanning tree of every component of $\tG$ in $\CC_F$. This step requires $O(\log\log{n}/\delta)$ MPC rounds on $\Ot(n^{1-\delta})$ machines of memory $\Ot(n^{\delta})$
	by Claim~\ref{clm:grow_components} (as $\card{E(\tG)} = \Ot(n)$ by construction). 
	Note that the components in $\CC_F$ 
	correspond to vertices of $H_F$ and hence by finding a spanning tree of $H_F$ we obtain a spanning tree of $\tG$. 
	
	By Claim~\ref{clm:H_F-diameter}, diameter of $H_F$ is only $O(1)$. Hence, by the algorithm in Claim~\ref{clm:low-diameter}, we can find a spanning tree of $H_F$ in only $O(1/\delta)$ rounds on machines of memory $O(n^{\delta})$. 
	Combining these trees, we obtain a spanning tree of $\tG$, finalizing the proof. 
\end{proof}

\subsection{Proof of Lemma~\ref{lem:third-step}: Connectivity on a Disjoint Union of Random Graphs}

We are now finally ready to prove Lemma~\ref{lem:third-step} using Lemma~\ref{lem:technical-random-cc}. 

\begin{proof}[Proof of Lemma~\ref{lem:third-step}]
	We perform the preprocessing step introduced in the beginning of the section to create the graph $\tG$ that is a graph which consists of $F$ copies of $\PG(n_i,\Delta \cdot s)$ where $n_i$ is the number of vertices in the connected
	component $G_i$ of $G$. By Proposition~\ref{prop:PG-expansion}, with high probability every connected component of $G$ which is sampled from $\PG$ (with degree $100\log{n}$) has mixing time of $\polylog{(n)}$. We can thus 
	apply Lemma~\ref{lem:second-step} to implement this step using $\Ot(n^{1-\delta})$ machines and $\Ot(n^{\delta})$ memory per machine in $O(\log\log{n}/\delta)$ rounds. 
		
	We then run the algorithm in Lemma~\ref{lem:technical-random-cc} on the whole graph. We can now analyze the algorithm in Lemma~\ref{lem:technical-random-cc} on the set of vertices belonging to each connected
	component $\tG_i$ of $\tG$ separately. It is immediate to verify that performance of algorithm in Lemma~\ref{lem:technical-random-cc} on each $\tG_i$ is only a function of $\tG_i$ and hence the correctness of the algorithm follows
	exactly as in Lemma~\ref{lem:technical-random-cc}. Hence, in $O(\log\log{n}/\delta)$ rounds, with high probability, we obtain a spanning tree of each $\tG_i$. We then assign a unique label to each spanning tree found and mark
	the vertices based on which spanning tree they belong to. Each label now corresponds to $V(\tG_i) = V(G_i)$, hence we can identify all connected components of $G$, finalizing the proof. The bound on the memory requirement
	and number of machines follows from Lemma~\ref{lem:technical-random-cc}.
\end{proof}

\section{Putting Everything Together}\label{sec:everything} 

We now put all components of our algorithms in the previous three sections together and prove the following theorem which formalize Theorem~\ref{res:main} in the introduction. 

\begin{theorem}\label{thm:main-m}
	There exists a randomized MPC algorithm that with high probability identifies all connected components of any given undirected $n$-vertex graph $G(V,E)$ with $m$ edges and a lower bound of
	$\lambda \in (0,1)$ on the spectral gap of any of its connected components.
	
	\noindent
	For any $\delta > 0$, the algorithm can be implemented with $O(\frac{1}{\lambda^2} \cdot m^{1-\delta} \cdot \polylog{(n)})$ machines each with $O(m^{\delta} \cdot \polylog{(n)})$ memory, and
	in ${O}(\frac{1}{\delta} \cdot \paren{\log\log{n} + \log{\paren{{1}/{\lambda}}}})$ MPC rounds. 
\end{theorem}
\begin{proof}
	We prove this theorem by applying the transformation steps in the previous sections to graph $G$. 
	
	\emph{\underline{Step 1.}} Let $G_1 := G$ with $n_1 := \card{V(G_1)}$ and $m_1:=\card{E(G_1)}$. 
	We apply Lemma~\ref{lem:first-step} to $G_1$ to obtain a $\Delta$-regular graph $G_2$ with the following properties (with high probability): there is a one-to-one correspondence between connected components 
	of $G_1$ and $G_2$, and each connected component of $G_2$ has mixing time $T_{\gamstar} = O(\log{n}/\lambda)$ with $\gamstar = n^{-10}$. Moreover, $n_2 := \card{V(G_2)} = O(m_1)$ and $\Delta = O(1)$.
	By identifying connected components of $G_2$, we immediately identify connected components of $G_1$. 
	
	This step can be implemented in $O(m^{1-\delta})$ machines with $O(m^{\delta})$ memory in $O(1/\delta)$ MPC rounds by Lemma~\ref{lem:first-step} (as $m_1 = m$).
	
	\emph{\underline{Step 2.}} We apply Lemma~\ref{lem:second-step} to $G_2$ with $T = T_{\gamstar}$ to (with high probability) obtain a graph $G_3$ such that $V(G_2) = V(G_3)$ and for any connected component $G_{2,i}$ on $n_{2,i}$ vertices,
	the induced subgraph of $G_3$ on vertices $V(G_{2,i})$, denoted by $G_{3,i}$, is a connected component of $G_{3}$ with distribution $\distribution{G_{3,i}}$ where $\tvd{\distribution{G_{3,i}}}{\PG(n_{2,i},100\log{n})} \leq n^{-10}$. 
	Identifying connected components of $G_2$ is equivalent to identifying connected components of $G_3$.
	
	This step can be implemented with $\Ot(n_2^{1-\delta})$ machines with $\Ot(n_2^{\delta})$ memory in $O(\log{T}/\delta)$ rounds by Lemma~\ref{lem:second-step}. Plugging in the value of these parameters, 
	we obtain that this step is implementable with $\Ot(m^{1-\delta})$ machines with $\Ot(m^{\delta})$ memory in $O(\frac{1}{\delta}\paren{\log\log{n} + \log{(1/\lambda)}})$ MPC rounds.
	
	\emph{\underline{Step 3.}} Let $n_3 = n_2$ be the number of vertices in $G_3$. We apply Lemma~\ref{lem:third-step} to $G_3$ to identify the connected components of $G$. 
	The distribution of each connected component $G_{3,i}$ of $G_{3}$ is $(n^{-8})$-close in total variation distance to $\PG(n_{2,i},100\log{n})$ ($n_{2,i} = \card{V(G_{3,i})}$).  Hence,
	by the guarantee of Lemma~\ref{lem:third-step} and Fact~\ref{fact:tvd-small}, with high probability we are able to identify connected components of the graph $G_3$. This allows us to identify connected components
	of $G_2$ and in turn $G_1 = G$. 
	
	This step can be implemented in $\Ot(n_3^{1-\delta})$ machines with $\Ot(n_3^{\delta})$ memory in $O(\log\log{n_3}/\delta)$ rounds by Lemma~\ref{lem:third-step}. Plugging in the value of these parameters, 
	we obtain that this step is implementable with $\Ot(m^{1-\delta})$ machines with $\Ot(m^{\delta})$ memory in $O(\log\log{n}/\delta)$ MPC rounds.

	This concludes the proof of the theorem.	
\end{proof}

\subsection*{Extension to Unknown Spectral Gaps}
A simple modification of our algorithm in Theorem~\ref{thm:main-m} allows for implementing it without having a prior knowledge of the spectral gap of each underlying connected component at the cost of slightly worse parameters. 
\begin{corollary}\label{cor:main-m}
	There exists a randomized MPC algorithm that for any $\delta > 0$, with high probability identifies all connected components of any given undirected $n$-vertex graph $G(V,E)$ with $m$ edges such that any connected component $G_i$ with 
	spectral gap $\lambda_2(G_i)$ (unknown to the algorithm) would be identified by the algorithm after ${O}(\frac{1}{\delta} \cdot \paren{\log\log{n}\cdot\log\log{(\frac{1}{\lambda_2(G_i)})} + \log{(\frac{1}{\lambda_2{(G_i)}})}})$ MPC rounds.
	
	\noindent	
	The algorithm requires $O(\frac{1}{\lambda^{2.1}} \cdot m^{1-\delta} \cdot \polylog{(n)})$ machines each with $O(m^{\delta} \cdot \polylog{(n)})$ memory, where $\lambda := \min_{i} \lambda_{2}(G_i)$ is the minimum spectral
	gap of any connected component of $G$. 
\end{corollary}
\begin{proof}

	We first choose $\lambda_1' = 1/2$ and run the algorithm in Theorem~\ref{thm:main-m} on $G$ with this choice of $\lambda'$. Let $\CC:= \set{C_1,\ldots,C_k}$ be the sets identified as connected components of $G$ by this algorithm. We note that
	algorithm in Theorem~\ref{thm:main-m} would always return a component-partition of $V(G)$ and hence if $u$ and $v$ belong to some $C_i \in \CC$, $u$ and $v$ also belong to the same connected component in $G$. However, it is possible
	that there exists some $u$ and $v$ such that $u,v \in G_i$ (for some particular connected component) but $C(u) \neq C(v)$ (as $\lambda'$ is not necessarily as small as spectral gap of $G_i$). It is easy to see that without loss of generality
	we can assume such $u$ and $v$ are neighbors in $G$. Hence, we can run a simple post-processing step to mark all components in $\CC$ which are a strict subset of some connected component of $G$, i.e., are ``growable'', and return
	the remaining components as connected components of $G$. This step can be trivially implemented in $O(1/\delta)$ MPC rounds on machines of memory $O(m^{\delta})$.

	We then recursively perform the above procedure by setting $\lambda'_j = (\lambda'_{j-1})^{1.1}$ in $j$-th recursion step on the vertices in marked components. 
	Fix any connected component $G_i$ of $G$. It is immediate that whenever $\lambda'_j \leq \lambda_{2}(G_i)$, the above procedure return this connected component (and hence would not mark it further). This means that  
	after $j^{\star} = O(\log\log{(\frac{1}{\lambda_2(G_i)}}))$ recursion steps, we have $\lambda'_{j^\star} \leq \lambda_{2}(G_i)$ and hence $G_i$ would be returned as a connected component. The total number of MPC rounds in these recursion steps is at most
	\begin{align*}
		O(\frac{1}{\delta}) \cdot \sum_{j=1}^{j^\star} \paren{\log\log{n} + \log{\frac{1}{\lambda'_{j}}}} &= O(\frac{1}{\delta})\cdot \Paren{\log\log{n} \cdot j^{\star} + \sum_{j=1}^{j^{\star}} (1.1)^{j}} \\
		&={O}(\frac{1}{\delta} \cdot \paren{\log\log{n}\cdot\log\log{(\frac{1}{\lambda_2(G_i)})} + \log{(\frac{1}{\lambda_2{(G_i)}})}}).
	\end{align*}
	Finally, it is easy to see that by the time $G_i$ is output, the algorithm has used $\Ot(\frac{1}{(\lambda'_{\jstar})^2} m^{1-\delta})$ machines, each with $\Ot(m^{\delta})$ memory; as $\lambda'_{\jstar} \geq \lambda_{2}(G_i)^{1.1}$, we obtain the
	final result. 
\end{proof}

\newcommand{\SC}{\ensuremath{\textnormal{\textsf{SublinearConn}}}\xspace}

\section{A Mildly Sublinear Space Algorithm for Connectivity}\label{sec:sublinear}

In this section, we present a simple algorithm for solving the sparse connectivity problem (in general, e.g., with no assumption on spectral gap, etc.) using $o(n)$ memory per-machine, proving Theorem~\ref{thm:sublinear} (restated below for convinence). 

\begin{theorem*}[Restatement of Theorem~\ref{thm:sublinear}]
	There exists an MPC algorithm that given any arbitrary graph $G(V,E)$ with high probability identifies all connected components of $G$ in $O(\log\log{n} + \log{\paren{\frac{n}{s}})}$ MPC rounds
	 on machines of memory $s = n^{\Omega(1)}$. 
\end{theorem*}

As a corollary of Theorem~\ref{thm:sublinear}, we have that $O(\log\log{n})$ rounds suffice to solve the connectivity problem even when memory per-machine is mildly sublinear in $n$, i.e., is $O(n/\polylog{(n)})$, and 
that as long as the memory per machine is $n^{1-o(1)}$, we can always improve upon the $O(\log{n})$-round classical PRAM algorithms for the connectivity problem on any arbitrary graph.  

The algorithm in Theorem~\ref{thm:sublinear} is a simple application of the toolkit we developed for proving our main result in Theorem~\ref{thm:main-m}, combined with the linear-sketching algorithm of Ahn~\etal~\cite{AhnGM12Linear} for graph connectivity. 
In particular, we use the following result from~\cite{AhnGM12Linear}. 

\begin{proposition}[\!\!\!\cite{AhnGM12Linear}]\label{prop:AGM}
	Let $H$ be any graph partitioned between $\card{V(H)}$ players such that each player receives all edges incident on a unique vertex in $V(H)$ (hence each edge is received by exactly two players). 
	There exists a randomized algorithm in which every player sends a message of size 	$O(\log^{3}{\card{V(H)}})$ bits to a central coordinator who can output all connected components of $H$ using only these messages with high probability. 
	The algorithm requires players to have access to $\polylog{(\card{V(H)})}$ \emph{shared random bits}. 
\end{proposition}

We are now ready to present the algorithm in Theorem~\ref{thm:sublinear}. We shall emphasize that unlike in our main result in Theorem~\ref{res:main}, to prove Theorem~\ref{thm:sublinear}, we do not need the full power
of essentially any of the steps we developed earlier in the paper and this result can be achieved using much simpler techniques as we show below. 

\begin{tbox}
$\SC(G)$. A mildly sublinear space algorithm for connectivity on a given graph $G$. 

	\algline	 

\begin{enumerate}
	\item Set $d := \frac{n \cdot (\log{n})^4}{s}$ and $t := \paren{d^3\cdot100\log{n}}$, and run $\SRW(G,t)$.  
	\item Create a graph $\tG$ from $G$ by connecting every vertex $v \in V(G)$ to all \emph{distinct} vertices visited in the random walk starting from $v$ computed in the previous step. 
	\item Run $\LE(\tG,d)$ and let $H$ be the graph obtained by contracting any component found by $\LE$ to a single vertex. 
	\item Remove self-loops and duplicate edges from $H$ and run the algorithm in Proposition~\ref{prop:AGM} on $H$ by using a dedicated machine to simulate a single player. 
\end{enumerate}
\end{tbox}

We are now ready to prove Theorem~\ref{thm:sublinear} using the \SC algorithm. 

\begin{proof}[Proof of Theorem~\ref{thm:sublinear}]

The correctness of the algorithm is based on the following observations: 
\begin{enumerate}
	\item Even though $G$ is \emph{not} a regular graph (as is needed in Theorem~\ref{thm:random-walk}), $\SRW(G,t)$ still finds a random walk of length $t$ from every vertex (by the discussions before 
	Observation~\ref{obs:multi-ind-layered-graph} and Lemma~\ref{lem:srw}). These walks are however \emph{not} independent of each other but we shall not need this property. 
	We further note that to actually find all vertices in the walk (and not only the final vertex) we use the \Mark~procedure defined previously.  
	
	\item A random walk of length $O(d^3\log{n})$ from any vertex would either visit all vertices in its connected component or at least $d$ distinct vertices with high probability. This follows 
	from a conjecture of Linial proven by Barnes and Feige in~\cite{BarnesF93} that states that the expected time to visit $N$ distinct vertices by a random walk is $O(N^3)$. 
	
	\item It follows from the previous part that the minimum degree of graph $\tG$ is at least $d$ with high probability. Even though $\tG$ is \emph{not} an almost-regular graph, it follows immediately
	from Claim~\ref{clm:LE-simple} and the proof of second part of Lemma~\ref{lem:LE-equal-size} that components found by $\LE$ contain all vertices of $\tG$ with high probability, i.e., form a component-partition of $G$. 
	
	\item It follows from the previous part that  $\card{V(H)} = O(n\log{n}/d) = O(s/\log^{3}{n})$ with high probability (we set sampling probability of leaders in $\LE$ to $\Theta(\log{n}/d)$ for this part as this parameter is already enough for the previous argument to work).  The correctness now follows from Proposition~\ref{prop:AGM}, as vertices in $H$ are components of $G$. 
\end{enumerate}

	To bound the number of rounds,  we need $O(\log{t}) = O(\log\log{n} + \log{\paren{\frac{n}{s}}})$ to implement $\SRW(G,t)$ by Claim~\ref{claim:detection_independent}, $O(1)$ rounds for $\LE$ by Claim~\ref{claim:leader_election}, and 
	$O(1)$ rounds for final step by Proposition~\ref{prop:AGM} and the fact that we can share $\polylog{(n)}$ random bits as well as removing duplicate edges in $O(1)$ rounds on machines with memory $n^{\Omega(1)}$. 
	To bound the memory per machine,  we need $n^{\Omega(1)}$ memory to implement $\SRW$ and $\LE$ and $O(\card{V(H)}\cdot\log^{3}{(n)})$ to implement the final step by Proposition~\ref{prop:AGM}. As argued, 
	$\card{V(H)} = O(s/\log^{3}(n))$ and hence $O(s)$ memory is sufficient for this step. This finalizes the proof. 
\Qed{Theorem~\ref{thm:sublinear}}

\end{proof}

\section{An Unconditional Lower Bound for Well-Connected Graphs}\label{sec:lb}

Our algorithmic results in this paper suggested that sparse connectivity is potentially ``much simpler'' on graphs with moderate expansion (i.e., with spectral gap $\lambda \geq 1/\polylog{(n)}$)
than on typical graphs. It is then natural, although perhaps \emph{too} optimistic, to wonder whether sparse connectivity on well-connected graphs is at all ``hard'' or not; for example, can we achieve an $O(\log\log{n})$-round
algorithm for finding well-connected components of a graph using only $\polylog{(n)}$ memory per machine in the MPC model, or perhaps directly in the PRAM model? Such a possibility would indeed 
imply that one does not need the ``full power'' of MPC algorithms (more local storage and computational power) to solve the sparse connectivity problem even on well-connected graphs. 
As we prove in this section, this is indeed not the case and full power of MPC algorithms are needed for connectivity even on well-connected graphs we considered in this paper. This 
supports the main message of our work on achieving \emph{truly improved} algorithms in the MPC model using the full power of this model.

We prove an \emph{unconditional} lower bound on the number of MPC rounds required to 
solve the connectivity problem on \emph{sparse} undirected graphs with a \emph{constant spectral gap}.  
More formally, we henceforth denote by $\ExpConn$ the decision \emph{promise} problem 
of determining connectivity on $n$-vertex graphs $G$, where in both cases (each connected component of) $G$ 
is guaranteed to be a sparse expander ($|E(G)|=O(n)$ and the spectral gap of each component is $\lambda_2 = \Omega(1)$). 

\begin{theorem}[Lower bound for expander connectivity] \label{thm:lb}
Every MPC algorithm for $\ExpConn$ with $s$ space per machine (and an \emph{arbitrary} 
number of machines) requires $r = \Omega(\log_s n)$ rounds of computation. This holds even against randomized 
MPC protocols with constant error probability. 
\end{theorem}

Theorem~\ref{thm:lb}, for example, suggests that any MPC algorithm with $\polylog{(n)}$ memory per machine for connectivity even on union of expander graphs requires $\Omega(\log{n}/\log\log{n})$ MPC rounds. In Remark~\ref{rem:erew-pram}, 
we further show a similar situation for (EREW) PRAM algorithms. 

We remark that by a result of~\cite{RoughgardenVW16}, the lower bound in Theorem~\ref{thm:lb} is asymptotically the best possible unconditional
lower bound short of proving that $\mathbf{NC}^1 \subsetneq \mathbf{P}$ which would be a major breakthrough in complexity theory.

Our lower bound is an adaptation of the argument in \cite{RoughgardenVW16}, who showed the same (asymptotic) 
lower bound for any (nontrivial) monotone graph property, albeit without the spectral gap nor the sparsity restrictions. 
They prove a general relationship between the round complexity of an MPC algorithm for 
computing a function $f$ and the \emph{approximate degree} of $f$ (see Theorem 3.5 and Proposition 2.7 in~\cite{RoughgardenVW16}). More formally, for a Boolean function 
$f: \{0,1\}^n \rightarrow \{0,1\}$, let 
$$\widetilde{deg}_\eps(f) := \min_{P: \{0,1\}^n \rightarrow \R}  \{ deg(P) \; \mid \; |f(x)-P(x)| \leq \eps \; \forall \; x\in \{0,1\}^n \} $$ 
denote the $\eps$-\emph{approximate degree}, i.e., the lowest degree of an ($n$-variate)  \emph{real} polynomial 
that uniformly $\eps$-approximates $f$ on the hypercube. The following proposition then follows from Corollaries 3.6 and 3.8 and Proposition 2.7 in~\cite{RoughgardenVW16}.

\begin{proposition}[\!\!\cite{RoughgardenVW16}] \label{prop:roughgarden}
If $f: \{0,1\}^n \rightarrow \{0,1\}$ is computable by an $r$-round randomized $\eps$-error MPC algorithm with space $s$ per machine, 
then $\widetilde{deg}_\eps(f) \leq s^{\Theta(r)}$.  
\end{proposition} 

By Proposition~\ref{prop:roughgarden}, proving Theorem~\ref{thm:lb} boils down to showing $\widetilde{deg}_\eps(\ExpConn) = n^{\Omega(1)}$, 
as this would imply an $r = \Omega(\log_s n)$ round lower bound for MPC algorithms for $\ExpConn$ with $s$ memory per machine. 

To prove such lower bounds on approximate degree of functions, \cite{RoughgardenVW16} 
further observed that it suffices to lower bound the \emph{deterministic decision tree} complexity $DT(f)$ (i.e., the \emph{query complexity}) of the underlying function, as it is known 
to imply a polynomially-related lower bound on its approximate degree. 

\begin{proposition}[Decision-tree complexity vs. approximate polynomial degree, \cite{Beals01,NS94}]\label{prop:dt-apd}
For any Boolean function $f: \{0,1\}^n \rightarrow \{0,1\}$, it holds that $\widetilde{deg}_{1/3}(f) \geq \Omega\left(DT(f)^{1/6}\right)$. 
\end{proposition}

We remark that the same bound applies to \emph{partial} functions defined on $\mathcal{D} \subseteq \{0,1\}^n$, using a straightforward 
generalization of the block-sensitivity measure and approximate degrees (see, e.g., Theorem 4.13 and the comment following it in \cite{Beals01}). 
It therefore remains to prove a lower bound on the deterministic decision tree complexity of $\ExpConn$, 
which is the content of the next lemma.  

\begin{lemma}\label{lem:dt-exp-conn}
$DT(\ExpConn) = \Omega(n/\log n)$.  
\end{lemma}

We shall use the following claim to construct our hard instances in the proof the lower bound in Lemma~\ref{lem:dt-exp-conn}. 
\begin{claim} \label{clm:expan}
There exists a collection of $k=\Omega(n)$ graphs 
$\cB := B_1,\ldots, B_k$ on the same set $V$ of $n$ vertices such that:
\begin{enumerate}
	\item Each $B_i$ is a $d$-regular graph with some fixed $d = O(1)$ and has spectral gap $\lambda_2(B_i) = \Omega(1)$.
	\item Any edge $e \in V \times V$ appears in at most $O(\log{n})$ different graphs $B_i \in \cB$. 
\end{enumerate}
\end{claim}

\begin{proof}
	We prove this claim using a probabilistic argument. Fix $d = 100$. Recall the definition of distribution $\FG_{n,d}$ from Section~\ref{sec:regularization}, i.e., the uniform distribution on $d$-regular graphs on $n$ vertices.
	We pick $k = n/100d$ graphs $\cB:= B_1,\ldots,B_k$ independently from $\FG_{n,d}$. By Corollary~\ref{cor:random_regular_graph} and a union bound, with high probability all these graphs are expanders with $\lambda_2(B_i) = \Omega(1)$. 
	
	Now consider any fixed edge $e \in V \times V$. For $i \in [k]$, define indicator random variables $X_i \in \set{0,1}$ where $X_i = 1$ iff $e \in B_i$. Define $X:= \sum_{i=1}^{k} X_i$ as the total number of graphs in $\cB$ to which $e$ belongs to.
	We have,
	\begin{align*}
		\Ex\bracket{X} = \sum_{i=1}^{k} \Ex\bracket{X_i} = k \cdot \Pr_{B \sim \FG_{n,d}}\paren{e \in B} = k \cdot \frac{2nd}{n^2} = \frac{1}{100}.
	\end{align*}
	As such, by Chernoff bound, the probability that $X \geq 4\log{n}$, i.e., $e$ appears in more than $4\log{n}$ graphs is at most $1/n^3$. Taking a union bound on all edges $e \in V \times V$, implies that with high probability no edge
	appears in more than $O(\log{n})$ different graphs in $\cB$. 
	
	Taking a union bound on the two events above, we obtain that with high probability, $\cB$ satisfies the requirement of the claim. This implies that in particular there should exists such collection $\cB$, finalizing the proof. \Qed{Claim~\ref{clm:expan}}
	
\end{proof}

We can now present the proof of Lemma~\ref{lem:dt-exp-conn}, which completes the entire proof of Theorem \ref{thm:lb}. 
\begin{proof}[Proof of of Lemma~\ref{lem:dt-exp-conn}]
Let $S,T$ be two disjoint subsets of vertices of size $n/2$ each and $G_S$ and $G_T$ be two $d$-regular expanders with $\lambda_2 = \Omega(1)$ on
vertices $S$ and $T$, respectively. Moreover, Let $\cB := \{B_i\}_{i=1}^k$ be the collection of $k$ expanders in Claim~\ref{clm:expan}. 
Note that the collection $\cB=\{B_1,\ldots, B_k\}$ is fixed 
in advance and known to the query algorithm (or equivalently the decision tree). Our final graph $G$ is going to contain \emph{at most one} of the 
graphs $B_i$, i.e., $G$ will either be $G_S\cup G_T$ (in the disconnected case), and other wise $G = G^{(i)} := (G_S \cup G_T \cup B_i)$ 
for \emph{some} $i\in [k]$ in the connected case. As such, Claim \ref{clm:expan} and the choice of $d$ guarantees that 
$G$ is a legitimate instance of the promise problem $\mathsf{ExpanderConn}_n$, i.e., that it is both sparse and has a 
constant spectral gap for each connected component as required.

We proceed with a standard adversarial argument. Without loss of generality, we assume that the query algorithm
only queries edges $e \in \bigcup_{i=1}^{k} E(B_i)$. Whenever the query algorithm queries an edge $e$ that belongs to 
$B_i$, the adversary declares that $B_i$, \emph{as well as all $B_j$'s for which $e\in B_j$} are \emph{not} present in $G$. 
Claim \ref{clm:expan} guarantees that at most $O(\log n)$ graphs $B_j$s are excluded for each one query. Therefore, the adversary can 
continue with the aforementioned strategy for $t=\Omega(k/\log n) = \Omega(n/d\log n) = \Omega(n/\log n)$ steps, and still there will be 
at least one \emph{unqueried} graph $B_{i^*}$. Therefore, if the query algorithm makes less than $t$ queries to $G$, the adversary can 
either declare $B_{i^*}$ is present or not (determining whether $G$ is connected or not), contradicting the algorithm's output in either case. 
\Qed{Lemma~\ref{lem:dt-exp-conn}}

\end{proof}

Theorem~\ref{thm:lb} now follows immediately from Lemma~\ref{lem:dt-exp-conn}, Proposition~\ref{prop:dt-apd}, and Proposition~\ref{prop:roughgarden}.

\begin{remark}[Extension to the PRAM model]\label{rem:erew-pram}
\emph{
	In Theorem~\ref{thm:lb} we proved the lower bound for $\ExpConn$  in the MPC model as our main focus in this paper is on this model after all. However, our proof of Theorem~\ref{thm:lb} implies
	that $\ExpConn$ requires $\Omega(\log{n})$ rounds in the (EREW) PRAM model as well. The reason is that our proof implies $\ExpConn$ is a \emph{critical} function of $\Omega(n/\log{n})$ variables: its output
	depends on the existence or non-existence of the $k = \Omega(n/\log{n})$ expanders graphs $B_1,\ldots,B_k$ (think of $\ExpConn$ as OR function of $k$ bits, each denoting whether the $i$-th expander $B_i$ is present in $G$ or not).  
	By results of~\cite{CookDR86,ParberryY91} computing such a function requires $\Omega(\log{n})$ rounds in the (EREW) PRAM model. 
}
\end{remark}

\subsection*{Acknowledgements}

We thank Alex Andoni, Soheil Behnezhad, Mahsa Derakhshan, Michael Kapralov, Sanjeev Khanna, and Krzysztof Onak for helpful discussions. We are additionally grateful to Soheil Behnezhad and Michael Kapralov for bringing up the question
of sparse connectivity using mildly sublinear space algorithms that prompted us to prove Theorem~\ref{thm:sublinear}.

\bibliographystyle{abbrv}
\bibliography{general}

\clearpage

\appendix

\section{Useful Concentration Bounds}\label{app:concentration}

We use the following standard version of Chernoff bound (see, e.g.,~\cite{ConcentrationBook}) throughout. 

\begin{proposition}[Chernoff bound]\label{prop:chernoff}
	Let $X_1,\ldots,X_n$ be independent random variables taking values in $[0,1]$ and let $X:= \sum_{i=1}^{n} X_i$. Then, for any $\eps \in (0,1)$, 
	\[ \Pr\paren{X \notin \range{(1 \pm \eps) \Ex\bracket{X}}} \leq 2 \cdot \exp\paren{-\frac{\eps^2 \cdot \Ex\bracket{X}}{2}}. \]
\end{proposition}

We also need the {method of bounded differences} in our proofs. A function $f(x_1,\ldots,x_n)$ satisfies the \emph{Lipschitz property} with constant $d$, iff 
for all $i \in [n]$, $\card{f(a) - f(a')} \leq d$, whenever $a$ and $a'$ differ only in the $i$-th coordinate. 

\begin{proposition}[Method of bounded differences] \label{prop:bounded-differences}
	If $f$ satisfies the Lipschitz property with constant $d$ and $X_1,\ldots,X_n$ are independent random variables, then, 
	\[ \Pr\paren{\card{f(X) - \Ex\bracket{f(X)}} > t} \leq \exp\paren{-\frac{2t^2}{n\cdot d^2}} \]
\end{proposition}
\noindent
A proof of this proposition can be found in~\cite{ConcentrationBook} (see Section 5).

\section{Balls and Bins Experiment}\label{app:balls-and-bins}

 We use the following standard balls and bins argument in our proofs. 
 
\begin{proposition}[Balls and Bins]\label{prop:balls-and-bins}
	Consider the process of throwing $N$ balls into $B$ bins where $N \leq \eps \cdot B$ for some parameter $\eps \in (0,1/100)$ such that each bin is chosen with probability in $\range{\paren{1 \pm \eps} \cdot {B^{-1}}}$. 
	Let $X$ denote the number of non-empty bins. Then,
	\begin{align*}
		\Pr\Paren{X \notin \range{\paren{1\pm 2\eps} \cdot N}} \leq\exp\paren{-\frac{\eps^2 \cdot N}{2}}. 
	\end{align*}
\end{proposition}

\begin{proof}
	Define an indicator random variable $X_i \in \set{0,1}$ for any $i \in [B]$, where $X_i = 1$ iff the $i$-th bin is non-empty. Clearly $X = \sum_{i=1}^{B} X_i$ denotes the number of non-empty bins. 
	As each bin is chosen (near) uniformly at random by a ball, we have that,
	\begin{align*}
		\Ex\bracket{X} = \sum_{i=1}^{B} \Ex\bracket{X_i} \in \range{B \cdot \Paren{1-\paren{1-\frac{(1 \pm \eps)}{B}}^{N}}} \in \range{(1\pm1.1\eps) \cdot N}.
		\tag{using the fact that $1-x \leq e^{-x} \leq 1-x + x^2/2$ for $x \leq 1$ and that $N/B \leq \eps$} 
	\end{align*}
	Random variables $X_1,\ldots,X_B$ are correlated and hence not amenable to a straightforward application of Chernoff bound. We instead use the method of bounded differences in Proposition~\ref{prop:bounded-differences} to prove
	the concentration of $X$ around $\Ex\bracket{X}$. 
	
	Define $N$ \emph{independent} random variables $Y_1,\ldots,Y_N$, where $Y_i$ denotes the index of the bin, the $i$-th ball is sent to. Define $f(Y_1,\ldots,Y_N)$ as the number non-empty bins (which is clearly only a function of $Y_1,\ldots,Y_N$). 
	We have $f(Y_1,\ldots,Y_N) = X$ and that $f$ is clearly $1$-Lipschitz as changing any $Y_i$ can only make one more bin empty or non-empty. As such, by Proposition~\ref{prop:bounded-differences},
	\begin{align*}
		\Pr\paren{\card{f(Y_1,\ldots,Y_N) - \Ex\bracket{f(Y_1,\ldots,Y_N)}} > (\eps/2) \cdot N} \leq  \exp\paren{-\eps^2 \cdot N/2}.
	\end{align*}
	As $f(Y_1,\ldots,Y_N) = X$ and $\card{\Ex\bracket{X} - N} \leq 1.1\eps N$, we have, 
	\begin{align*}
		\Pr\paren{X \notin \range{\paren{1\pm 2\eps} \cdot N}} \leq  \exp\paren{-\eps^2 \cdot N/2},
	\end{align*}
	finalizing the proof. 
\end{proof}

\section{Replacement and Zig-Zag Products on Non-Regular Graphs}\label{app:replacement-zigzag}

Let $G$ be a graph on $n$ vertices $v_1,\ldots,v_n$ with degree $d_v$ for $v \in V(G)$, and $\HH$ be a family of $n$ $d$-regular graphs $H_1,\ldots,H_n$ where $H_v$ is supported on $d_v$ vertices (we assume $d_v \geq d$ for all $v \in V(G)$).
We construct the \emph{replacement} product $G \replacement \HH$ as follows: 
\begin{itemize}
	\item Replace the vertex $v$ of $G$ with a copy of $H_v$ (henceforth referred to as a \emph{cloud}). For any $i \in H_v$, we use $(v,i)$ to refer to the $i$-th vertex of the cloud $H_v$. 
	
	\item Let $(u,v)$ be such that the $i$-th neighbor of $u$ is the $j$-th neighbor of $v$. Then there exists an edge between vertices $(u,i)$ and $(v,j)$ in $G \replacement \HH$. Additionally, for any $v \in V(G)$, if there exists an edge $(i,j) \in H_v$,
	then there exists an edge $((v,i),(v,j))$ in $G \replacement \HH$. 
\end{itemize}
It is easy to see that the replacement product $G \replacement \HH$ is a $(d+1)$-regular graph on $2m$ vertices where $m$ is the number of edges in $G$. 

We prove Proposition~\ref{prop:replacement-product} from Section~\ref{sec:regularization} in this section. For convenience, we repeat the 
statement of the proposition again. 

\begin{proposition*}[Proposition~\ref{prop:replacement-product} in Section~\ref{sec:regularization}]
	Suppose $\lambda_2(G) \geq \lambda_G$ and all $H_v \in \HH$ are $d$-regular with $\lambda_2(H_v) \geq \lambda_H$. Then, 
	$\lambda_2(G \replacement \HH) = \Omega\paren{d^{-1} \cdot \lambda_G \cdot \lambda_H}$
\end{proposition*}

To prove Proposition~\ref{prop:replacement-product}, it would be more convenient to consider
 a slightly more involved graph product, i.e., the so-called \emph{zig-zag product}. The zig-zag product of a graph $G$ with a family $\HH$ of graphs supported on vertex-degrees of $G$
is a graph with the same set of vertices as the replacement product $G \replacement \HH$ with some additional edges. Intuitively, the new edges connect endpoints of the length-$3$ paths in $G \replacement \HH$ which consist of taking one edge inside a cloud, one edge between two neighboring clouds, and one edge inside the next cloud (hence the name ``zig-zag''). 

Formally, for a graph $G$ on $n$ vertices $v_1,\ldots,v_n$ with degree $d_v$ for all $v \in V(G)$, and a family $\HH$ of $n$ $d$-regular graphs $H_1,\ldots,H_n$ where $H_v$ is supported on $d_v$ vertices (we assume $d_v \geq d$ for all $v \in V(G)$),
we construct the \emph{zig-zag} product $G \zigzag \HH$ as follows: 

\begin{itemize}
	\item The vertex-set of $G \zigzag \HH$ is the same as $G \replacement \HH$. 
	\item A vertex $(u,i)$ is connected to $(v,j)$ if there exist $k$ and $\ell$ such that the edges $((u,i),(u,k))$, $((u,k),(v,\ell))$, and $((v,\ell),(v,j))$ belong to $G \replacement \HH$. 
\end{itemize}
It is straightforward to verify that the graph $G \zigzag \HH$ is a $d^2$-regular graph on $2m$ vertices ($m$ is number of edges in $G$). 

The following proposition asserts that the spectral gap is preserved under zig-zag product.

\begin{proposition}[cf. \cite{ReingoldVW00,ReingoldTV06}]\label{prop:zigzag-product}
	Suppose $\lambda_2(G) \geq \lambda_G$ and all $H_v \in \HH$ are $d$-regular with $\lambda_2(H_v) \geq \lambda_H$. Then, 
	$\lambda_2(G \replacement \HH) \geq {\lambda_G \cdot \lambda_H^2}.$
\end{proposition}

Similar to Proposition~\ref{prop:replacement-product}, Proposition~\ref{prop:zigzag-product} was also first proved in~\cite{ReingoldVW00} for the case when $G$ is also a $D$-regular graph and all copies of $H_1,\ldots,H_N$ are the same $d$-regular graph on 
$D$ vertices (this is the case for all proofs of this proposition that we are aware of). For completeness, we provide a proof of this proposition when $G$ is not regular. Extending Proposition~\ref{prop:zigzag-product} to 
Proposition~\ref{prop:replacement-product} can be done immediately using known results. 

\subsection*{Preliminaries}
We start with some simple preliminaries needed in the proof of Proposition~\ref{prop:zigzag-product}.

For any $n \times n$ matrix $M$, we use $\norm{M}_2$ to denote the spectral norm of $M$ defined as: 
\begin{align*}
	\norm{M}_2 := \max_{x \in \IR^{n} \wedge \norm{x}_2 = 1} \norm{M \cdot x}_2,
\end{align*}
where $\norm{\cdot}_2$ for a vector is defined in the standard way. 

The following standard proposition relates eigenvalues of a random walk matrix with spectral gap of the underlying graph. 
\begin{proposition}\label{prop:W-bound}
	Let $W$ be a random walk matrix of a $d$-regular graph $H$ and $\mu_2(W)$ be the second largest eigenvalue of $W$. Then, $\mu_2(W) = 1-\lambda_2(H)$.
\end{proposition}
\begin{proof}
		Let $\LL$ denote the normalized Laplacian of $H$ and $A$ be its adjacency matrix. Recall that $W = D^{-1} \cdot A = d^{-1} \cdot A$ (as $H$ is $d$-regular) and hence 
		$W =  I - \LL_G$. We thus have $\mu_2(W) = 1-\lambda_2(\LL) = 1-\lambda_2(H)$. 
\end{proof}

We have the following characterization of the second largest eigenvalue of a symmetric matrix. 
\begin{proposition}\label{prop:W-mu2}
	Let $W$ be a random walk matrix of a $d$-regular $n$-vertex graph and $\ones_{n}$ be the $n$-dimensional vector of all ones: then 
	\begin{align*}
		\mu_2(W) = \max_{x \perp \ones_{n} ~\wedge~ \norm{x} = 1} \norm{W \cdot x}_2.
	\end{align*}
\end{proposition}
\begin{proof}
	Follows because $W$ is symmetric and $\ones_{n}$ is an eigenvector of $W$ corresponding to its largest eigenvalue (which is $1$).
\end{proof}

We further use the following proposition due to Rozenman and Vadhan~\cite{RozenmanV05} that proves a decomposition for a random walk matrix of a regular graph. 
\begin{proposition}[\!\!\cite{RozenmanV05}]\label{prop:matrix-decomposition}
	Let $H$ be a $d$-regular $N$-vertex graph and $W_H$ denotes its random walk $N \times N$ matrix. Let $J_N$ be a $N \times N$ matrix with all entries $1/N$. Then $W_H = \lambda_2(H) \cdot J_N + \paren{1-\lambda_2(H)} \cdot C$, where
	$\norm{C} \leq 1$. 
\end{proposition}

\subsection*{Proof of the Non-Regular Zig-Zag Product: Proposition~\ref{prop:zigzag-product}}

We follow the approach of~\cite{ReingoldTV06} in proving this proposition which itself was inspired by~\cite{RozenmanV05}. 

 \begin{proof}[Proof of Proposition~\ref{prop:zigzag-product}]
 	Let $W_G$, and $W_{H_1},\ldots,W_{H_n}$ be the random walk matrices of $G$ and clouds $H_1,\ldots,H_n$, respectively, and $W$ be the random walk matrix of $G \zigzag \HH$. 

	By Proposition~\ref{prop:W-bound}, we need to bound $\mu_2(W)$ to prove the final result. Consider the following two auxiliary matrices: 
	\begin{itemize}
		\item $B$: a block-diagonal $2m \times 2m$ matrix with $n$ blocks corresponding to vertices $v \in V(G)$ where the $v$-th block is the $d_v \times d_v$ matrix $W_{H_v}$. 
		\item $P$: a $2m \times 2m$ matrix corresponding to the matching that connects vertex $(u,i)$ to vertex $(v,j)$ in $G \replacement \HH$ whenever $(u,v)$ is the $i$-th edge of $u$ and $j$-th edge of $v$ in $G$. 
	\end{itemize}
	
	Recall the construction of the zig-zag product $G \zigzag \HH$: we take a step in some cloud (corresponding to move according to $B$), take a step between the clouds (corresponding to move according to $P$), and then take another 
	step inside a cloud (again corresponding to move according to $B$). It is then easy to see that $W = B P B$. 
	
	Recall that while $G$ is not necessarily regular, the graphs in $\HH$ are all $d$-regular. Hence, we can apply Proposition~\ref{prop:matrix-decomposition} to each $H_v \in \HH$. We define two more matrices based on this:
	
	\begin{itemize}
		\item {${J}$}: a block-diagonal $2m \times 2m$ matrix with $n$ blocks where the $v$-th block is a $d_v \times d_v$ matrix $J_v := \frac{1}{d_v} \cdot \ones_{d_v \times d_v}$; here,  $\ones_{d_v \times d_v}$ is the 
		matrix of all ones. 
				
		\item $C$: a block-diagonal $2m \times 2m$ matrix with $n$ blocks where the $v$-th block is a $d_v \times d_v$ matrix defined as follows: Apply Proposition~\ref{prop:matrix-decomposition} to $H_v$ to get
		$W_{H_v} = \lambda_2(H_v) \cdot J_v + \paren{1-\lambda_2(H_v)} \cdot C_v$ where $C_v$ is some $d_v \times d_v$ matrix with $\norm{C_v} \leq 1$. Let the matrix in the $v$-th block of $C$ be $C_v$.
	\end{itemize}
	Without loss of generality, we assume that $\lambda_2(H_v) = \lambda_H$ in the following argument (as opposed to $\lambda_2(H_v)\geq \lambda_H$). 
	Using above two matrices, we can write $B = \lambda_H \cdot J + \paren{1-\lambda_H} \cdot C$. Moreover, $\norm{C} \leq 1$ as well. Consequently, 
	\begin{align}
		W &= BPB = \Paren{\lambda_H \cdot J + \paren{1-\lambda_H} \cdot C} \cdot P \cdot \Paren{\lambda_H \cdot J + \paren{1-\lambda_H} \cdot C} \notag \\
		&= \lambda_H^2 \cdot JPJ + \paren{1-\lambda^2_H} \cdot \widetilde{C}, \label{eq:decomposition}
	\end{align} 
	where $\widetilde{C}$ is another matrix with $\norm{\widetilde{C}} \leq 1$. We use this equation to bound $\mu_2(W)$. 
	
	Recall that $G \zigzag H$ is a $d^2$-regular graph (thus $W$ is symmetric). Let $\ones_{2m}$ be a $2m$-dimensional vector of all ones. By Proposition~\ref{prop:W-mu2}, we have, 
	\begin{align*}
		\mu_2(W) = \max_{x \perp \ones_{2m} ~\wedge~ \norm{x} = 1} \norm{W \cdot x}_2.
	\end{align*}
	Fix any $x$ with $\norm{x} = 1$ such that $x \perp \ones_{2m}$.  By Eq~(\ref{eq:decomposition}), we have, 
	\begin{align}
		\norm{W \cdot x} &\Eq{Eq~(\ref{eq:decomposition})} \norm{\Paren{\lambda_H^2 \cdot JPJ + \paren{1-\lambda^2_H} \cdot \widetilde{C}}\cdot x} \leq \lambda_H^2 \cdot \norm{JPJ \cdot x} + \paren{1-\lambda^2_H}, \label{eq:suffices}
	\end{align}
	as norm $\widetilde{C}$ is at most $1$. Hence, it suffices to bound $\norm{JPJ \cdot x}$ to finalize the proof, which we do in the following claim. We point out that the
	following claim is the main part where we defer from the previous proofs of zig-zag product that assumed $G$ is also regular. 
	\begin{claim}\label{clm:vector-bound}
		$\norm{J P J \cdot x} \leq (1-\lambda_G)$ for all $x \perp \ones_{2m}$ with $\norm{x} = 1$. 
	\end{claim}
	\begin{proof}
		Let $A_G$ be the adjacency matrix of $G$. We define an $n \times n$ matrix $N_G$ where for all $(u,v)$: 
		\[
		(N_G)_{u,v} = \frac{1}{\sqrt{d_u \cdot d_v}} (A_G)_{u,v}.
		\]
		Recall that $J$ is a block-diagonal matrix where block $v$ is $\frac{1}{d_v} \cdot \ones_{d_v \times d_v}$. 
		Moreover, any entry $(u,i),(v,j)$ is $1$ in $P$ iff
		$(u,v)$ is the $i$-th edge of $u$ and $j$-th edge of $v$ in $G$. Using this, we can write any entry $(u,i),(v,j)$ of $JPJ$ as: 
		\begin{align*}
			\paren{J P J}_{(u,i),(v,j)} = \frac{1}{d_u \cdot d_v} \cdot (A_G)_{u,v} .
		\end{align*}
		
		We write $x = [x_1,\ldots,x_n]^{T}$ where $x_v$ is a $d_v$-dimensional vector whose entries we denote by $x_{v,j}$ for $j \in [d_v]$. We have,
		\begin{align*}
			\norm{JPJ x}^2 &= \sum_{(u,i)} \paren{\sum_{(v,j)} \paren{(JPJ)_{(u,i),(v,j)} \cdot x_{v,j}}}^2 = \sum_{(u,i)} \paren{\sum_{(v,j)} \frac{1}{d_u \cdot d_v} \cdot (A_G)_{u,v} \cdot x_{v,j}}^2 \\
			&= \sum_{u} d_u \cdot \paren{\sum_{v}\sum_{j=1}^{d_v} \frac{1}{d_u \cdot d_v} \cdot (A_G)_{u,v} \cdot x_{v,j}}^2.
		\end{align*}
		Define the $n$-dimensional vector $y$ where $y_v := \sum_{j=1}^{d_v}\frac{1}{\sqrt{d_v}} \cdot x_{v,j}$. Plugging in this value in the above bound, we have, 
		\begin{align*}
			\norm{JPJ x}^2 = \sum_{u} \paren{\sum_{v}\frac{1}{\sqrt{d_u \cdot d_v}} \cdot (A_G)_{u,v} \cdot y_v}^2 = \sum_{u} \paren{\sum_{v} (N_{G})_{u,v}\cdot y_v}^2 = \norm{N_{G} \cdot y}^2.
		\end{align*}
		
		Finally, we argue that $\norm{N_G \cdot y}^2 \leq (1-\lambda_G)^2$, which concludes the proof. We start with the following proposition relating eigenvalues of $N_G$ and $W_G$. 		
		\begin{proposition}\label{prop:N_G-M_G}
			Let $z$ be an eigenvector of $N_G$ with eigenvalue $\mu$. Then $D_G^{-1/2} \cdot z$ is an eigenvector of $W_G$ with eigenvalue $\mu$.
		\end{proposition}
		\begin{proof}
		 It is immediate to see that $W_G = D_G^{-1} \cdot A_G$ and $N_G = {D_G}^{-1/2} \cdot A_G \cdot {D}^{-1/2}$. This implies that $W_G = {D_G}^{-1/2} \cdot N_G \cdot {D}^{1/2}$. Consider $z' = {D_G}^{-1/2}\cdot z$. We have, 
			 \begin{align*}
			 	W_G \cdot z' = W_G \cdot {D_G}^{-1/2} \cdot z =  {D_G}^{-1/2} \cdot N_G \cdot z = {D_G}^{-1/2} \cdot \mu \cdot z = \mu \cdot z'.
			 \end{align*}
			 This concludes the proof. 
		\end{proof}
		We also have, 
		\begin{align*}
			\norm{y}^2 = \sum_{v}\paren{\sum_{j=1}^{d_v} \frac{1}{\sqrt{d_v}} x_{v,j}}^2 \leq \sum_{v} \paren{\sum_{j=1}^{d_v} \frac{1}{d_v}} \cdot \paren{\sum_{j=1}^{d_v}x_{v,j}^2} = \norm{x}^2 = 1,
		\end{align*}
		where the inequality is by Cauchy-Schwarz. We further have $y \perp D_G^{1/2} \cdot \ones_{n}$ as
		\begin{align*}
			\sum_{v} \sqrt{d_v} \cdot y_v= \sum_{v} \sqrt{d_v} \cdot \sum_{j=1}^{d_v} \frac{1}{\sqrt{d_v}} \cdot x_{v,j} = \sum_{(v,j)} x_{(v,j)} = 0,
		\end{align*}
		as $x \perp \ones_{2m}$. Finally, by Proposition~\ref{prop:N_G-M_G}, $D_G^{1/2} \cdot \ones_{n}$ is an eigenvector of $N_G$ corresponding to its maximum eigenvalue (as $D_G$ is an eigenvector of $W_G$ with maximum eigenvalue). 
		As such, 
		\begin{align*}
			\norm{N_G \cdot y} \leq \max_{\norm{y'}=1~\wedge~y' \perp D_G^{1/2} \cdot \ones_{n}} \norm{N_G \cdot y'} = \mu_2(N_G), 
		\end{align*}
		where the equality above is because $N_G$ is a symmetric matrix (even though  $G$ may not be regular) and $D_G^{1/2} \cdot \ones_{n}$ is an eigenvector of $N_G$ with largest eigenvalue (the proof of the equation
		is then identical to the proof of Proposition~\ref{prop:W-mu2}).

		Finally, $N_G = I - \LL_G$ and hence $\mu_2(N_G) = 1-\lambda_2(G) = 1-\lambda_G$, finalizing the proof.
		\Qed{Claim~\ref{clm:vector-bound}}
		
	\end{proof}
	
	To conclude, by Eq~(\ref{eq:suffices}),
	\begin{align*}
		\mu_2{(W)} &= \norm{W \cdot x} \Leq{Eq~(\ref{eq:suffices})} \lambda_H^2 \cdot \norm{JPJ \cdot x} + \paren{1-\lambda^2_H} \\
		&\!\!\!\!\!\!\!\!\Leq{Claim~\ref{clm:vector-bound}} \lambda_H^2 \cdot \paren{1-\lambda_G} + \paren{1-\lambda^2_H} \\
		&= 1- \lambda_G \cdot \lambda_H^2.
	\end{align*}
	By Proposition~\ref{prop:W-bound}, we have $\lambda_2(G \zigzag \HH) = 1-\mu_2(W) \geq \lambda_G \cdot \lambda_H^2$. 
	\Qed{Proposition~\ref{prop:zigzag-product}}	
	
 \end{proof}
 
 \subsection*{Proof of Non-Regular Replacement Product: Proposition~\ref{prop:replacement-product}}
 We now extend the proof of zig-zag-product in Proposition~\ref{prop:zigzag-product} to the replacement product and prove Lemma~\ref{prop:replacement-product}.
 We note that unlike the proof in the previous section, this part follows directly from the proof of~\cite{ReingoldVW00} and is only provided for completeness. 
 \begin{proof}[Proof of Proposition~\ref{prop:replacement-product}]
 	Let $W_r$ be the random walk matrix of the replacement product $G \replacement \HH$. Define $2m \times 2m$ matrices $B$ and $J$ as in Proposition~\ref{prop:zigzag-product}. One can verify that, 
	\begin{align*}
		W_r = \frac{{P + d \cdot B}}{{d+1}}. 
	\end{align*}
	Consequently, 
	\begin{align*}
		W_r^3 = \Paren{\frac{{P + d \cdot B}}{{d+1}}}^3 = \frac{d^2}{(d+1)^3} \cdot BPB + \paren{1-\frac{d^2}{(d+1)^3}} \cdot C,
	\end{align*}
	where $\norm{C} \leq 1$. Note that $BPB$ is the random walk matrix $W_z$ of the zig-zag product $G \zigzag \HH$ (see the proof of Proposition~\ref{prop:zigzag-product}). 
	We now bound $\mu_2(W_r^3)$ as follows. By Proposition~\ref{prop:W-mu2}, we have,  
		\begin{align*}
		\mu_2(W_r^3) = \max_{x \perp \ones_{2m} \wedge x \perp \ones_{2m}} \norm{x}. 
	\end{align*}
	as $W_r^3$ is a random walk matrix of some regular graph and hence its stationary distribution is $\ones_{2m}$. Fix any $x \perp 1$ with $\norm{x} = 1$; 
	\begin{align*}
		\norm{W_r^3 \cdot x} \leq \frac{d^2}{(d+1)^3} \cdot \norm{BPB x} + \paren{1-\frac{d^2}{(d+1)^3}} \cdot \norm{C} \leq \frac{d^2}{(d+1)^3} \cdot \mu_2(W_z) + \paren{1-\frac{d^2}{(d+1)^3}}. 
	\end{align*}
	Define $\overline{d} := \frac{d^2}{(d+1)^3}$. 
	As we bounded $\mu_2(W_z) \leq 1-\lambda_G \cdot \lambda_H^2$ in Proposition~\ref{prop:zigzag-product}, by above equation we have, 
	\begin{align}
		\mu_2(W_r^3) \leq \overline{d} \cdot \paren{1-\lambda_G \cdot \lambda_H^2} + \paren{1-\overline{d}} = 1-\overline{d} \cdot \lambda_G \cdot \lambda_H^2. \label{eq:final-damn}
	\end{align}
	As eigenvalues of powers of matrices are the respective powers of the original eigenvalues, and by Proposition~\ref{prop:W-bound},
	\begin{align*}
		\lambda_2(G \replacement \HH) &\Eq{Proposition~\ref{prop:W-bound}} 1-\mu_2(W_r) = 1- (\mu_2(W^3_r))^{1/3} \\
		&\qquad\!\!\Geq{Eq~(\ref{eq:final-damn})} 1- \paren{1-\overline{d} \cdot \lambda_G \cdot \lambda_H^2}^{1/3} \\
		&\qquad~~\geq \frac{1}{6} \cdot \overline{d} \cdot \lambda_G \cdot \lambda_H^2 \tag{as $1-x \leq e^{-x} \leq 1-x/2$ for $x \in (0,1)$} \\
		&\qquad~~= \Omega(d^{-1} \cdot \lambda_G \cdot \lambda_H^2),
	\end{align*}
	by the choice of $\overline{d}$. This concludes the proof. 
 \end{proof}

\section{Omitted Proofs of Simple Properties of Random Graphs}\label{app:random-graph}

For completeness, we provide short proofs for the property of random graphs $\PG$ stated in Section~\ref{sec:random-graph}.  These propositions are restated here for convenience. 

\begin{proposition*}[\textbf{Almost-regularity}]
	Suppose $d \geq 4\log{n}/\eps^2$ for some parameter $\eps \in (0,1)$. A graph $G \sim \PG(n,d)$ is an $\range{\paren{1 \pm \eps}d}$-almost-regular with high probability. 
\end{proposition*}
 \begin{proof}
 	Follows from a direct application of Chernoff bound on degree of each vertex and taking a union bound on all $n$ vertices. 
 \end{proof}
 
\begin{proposition*}[\textbf{Connectivity}]
	A graph $G \sim \PG(n,d)$ for $d \geq c \log{n}$ is connected with probability at least $1-1/n^{(c/4)}$. 
\end{proposition*}
\begin{proof}
	Let $S \subseteq G$ with $\card{S} \leq n/2$ and consider the cut $(S,V(G) \setminus S)$ in $G$. For $G$ to be disconnected, at least one such cut $S$ should contain no edges in $G$. 
	Any of the $d/2$ neighbors chosen for each vertex in $S$ in the process of creating $G$, 
	has at least $1/2$ probability of being in $V(G) \setminus S$, simply since $\card{V(G) \setminus S} \geq n/2$. As the choice of all these $\card{S} \cdot d/2$ vertices are independent, the probability
	that no edge crosses this particular cut is at most $2^{-\card{S} \cdot d/2}$. We now take a union bound on all possible choices for $S$: 	
	\begin{align*}
		\Pr\paren{\text{$G$ is not connected}} \leq \sum_{k=1}^{n/2} {{n} \choose {k}} \cdot 2^{-k \cdot (c/2) \cdot \log{n}} \leq \sum_{k=1}^{n/2} n^{k} \cdot 2^{-k \cdot (c/2) \cdot \log{n}} \leq 2^{-(c/4) \cdot \log{n}} = n^{-(c/4)}, 
	\end{align*}
	concluding the proof. 
\end{proof}
\begin{proposition*}[\textbf{Expansion}]
	Suppose $G \sim \PG(n,d)$ for $d \geq c \log{n}$. Then, with probability at least $1-1/n^{(c/4)}$: 
	\vspace{-5pt}
	\begin{enumerate}
		\item For any set $S \subseteq V(G)$ the neighborset $N(S)$ of $S$ in $G$ has size $\card{N(S)} \geq \min\set{2n/3,d/12 \cdot \card{S}}$.
		\item The mixing time of $G$ is $T_{\gamma}(G) = O(d^2\cdot \log{(n/\gamma)})$ for any $\gamma < 1$.
	\end{enumerate}
\end{proposition*}
\begin{proof}
	The proof of the first part is similar to the previous proposition; when picking a neighbor for any single vertex in $S$, there is at least $1/3$ chance that this vertex is not chosen by any of the previous vertices in $S$. A simple application of Chernoff bound
	plus a union bound on all cuts then concludes the proof exactly as in previous proposition and we omit the details.
	
	Let $\lambda_2(G)$ denote the spectral gap of $G$. The first part of the proposition already implies that vertex expansion of $G$ for any set $S$ of size up to $n/2$ is $\Omega(1)$. By the well-known connection between
	 vertex- and spectral-expansion (see, e.g.~\cite{hoory2006expander,Chung97}), this implies that $\lambda_2(G) = \Omega(1/d^2)$. The bound on mixing time now follows from Proposition~\ref{prop:mixing-spectral}.
\end{proof}

\end{document}